\documentclass[10pt,a4paper]{article}
\usepackage[utf8]{inputenc}
\usepackage[T1]{fontenc}

\usepackage[left=3cm, right=3cm, top=2cm, bottom=2cm]{geometry}
\usepackage{pdflscape}
\usepackage{palatino}

\usepackage{authblk}

\usepackage{amsmath,
            mathtools}
\usepackage{amsfonts}
\usepackage{amssymb}
\usepackage{amsthm}

\usepackage{mathdots}
\usepackage{yhmath}
\usepackage{cancel}
\usepackage[dvipsnames]{xcolor}
\usepackage{array}
\usepackage{multirow}
\usepackage{tabularx}
\usepackage{extarrows}
\usepackage{booktabs}
\usepackage{float} 

\usepackage[colorlinks = true, allcolors = blue]{hyperref}

\usepackage{todonotes}
\usepackage{tikz,
            quantikz}
\usepackage[compat=0.3]{yquant}
\usepackage{thmtools}

\usepackage[nameinlink]{cleveref}
\Crefname{section}{Section}{Sections}
\Crefname{Algo}{Algorithm}{Algorithms}
\Crefname{Def}{Definition}{Definitions}
\Crefname{Thm}{Theorem}{Theorems}
\Crefname{Pro}{Problem}{Problems}
\Crefname{Prop}{Proposition}{Propositions}
\Crefname{Lem}{Lemma}{Lemmas}

\newtheorem{Algo}{Algorithm}[section]
\newtheorem{Thm}[Algo]{Theorem}
\newtheorem{Prop}[Algo]{Proposition}
\newtheorem{Coro}[Algo]{Corollary}
\newtheorem{Lem}[Algo]{Lemma}

\theoremstyle{definition}
\newtheorem{Def}[Algo]{Definition}
\newtheorem{Pro}[Algo]{Problem}

\theoremstyle{remark}
\newtheorem{Rem}[Algo]{Remark}
\newtheorem{Ex}[Algo]{Example}

\numberwithin{equation}{section}

\renewcommand{\ket}[1]{|{#1}\rangle}
\renewcommand{\bra}[1]{\langle{#1}|}
\newcommand{\bk}[2]{\langle{#1}|{#2}\rangle}
\newcommand{\kb}[2]{\ket{#1}\bra{#2}}
\renewcommand{\proj}[1]{\kb{#1}{#1}}
\renewcommand{\epsilon}{\varepsilon}
\newcommand{\hepsilon}{\hat{\varepsilon}}
\newcommand{\x}{\otimes}
\newcommand{\ct}{^{\dagger}}

\newcommand{\Z}{\mathbb{Z}}
\newcommand{\C}{\mathbb{C}}

\newcommand{\R}{\mathbb{R}}
\newcommand{\F}{\mathbb{F}}

\newcommand{\Ghat}{\hat{G}}
\newcommand{\CG}{\C^G}
\newcommand{\CGhat}{\C^{\Ghat}}
\newcommand{\U}[1]{\mathrm{U}(#1)}

\renewcommand{\Re}{\mathop{\mathrm{Re}}\nolimits}

\usepackage{mathtools}
\DeclarePairedDelimiter{\set}{\lbrace}{\rbrace}
\DeclarePairedDelimiter{\abs}{\lvert}{\rvert}
\DeclarePairedDelimiter{\norm}{\lVert}{\rVert}

\DeclarePairedDelimiter{\of}{\lparen}{\rparen}

\DeclareMathOperator{\Aut}{Aut}
\DeclareMathOperator{\poly}{poly}
\DeclareMathOperator{\rot}{rot}
\DeclareMathOperator{\Tr}{Tr}

\newcommand{\mx}[1]{\begin{pmatrix}#1\end{pmatrix}}

\newcommand{\hf}{\hat{f}}
\newcommand{\hg}{\hat{g}}
\newcommand{\hG}{\hat{G}}
\newcommand{\hr}{\hat{r}}
\newcommand{\hR}{\hat{R}}
\newcommand{\hA}{\hat{A}}
\newcommand{\halpha}{\hat{\alpha}}

\title{Hidden shift problem for complex functions}

\author[1]{Serge Adonsou}
\author[2]{Peter Bruin}
\author[3,4]{Maris Ozols}
\author[2,3]{Joppe Stokvis}

\affil[1]{Department of Mathematics and Statistics, University of Guelph, Canada}
\affil[2]{Mathematical Institute, Leiden University, Netherlands}
\affil[3]{QuSoft, Amsterdam, Netherlands}
\affil[4]{Institute for Logic, Language and Computation, and Korteweg-de Vries Institute for Mathematics, University of Amsterdam, Netherlands}
\affil[ ]{Email addresses: \href{mailto: sadonsou@uoguelph.ca}{sadonsou@uoguelph.ca}, \href{mailto: p.j.bruin@math.leidenuniv.nl}{p.j.bruin@math.leidenuniv.nl}, \href{mailto: marozols@gmail.com}{marozols@gmail.com}, \href{mailto: j.a.stokvis@math.leidenuniv.nl}{j.a.stokvis@math.leidenuniv.nl}}
\begin{document}
\maketitle

\begin{abstract}
    We study quantum algorithms for the hidden shift problem of complex scalar- and vector-valued functions on finite abelian groups.
    Given oracle access to a shifted function and the Fourier transform of the unshifted function, the goal is to find the hidden shift.
    We analyze the success probability of our algorithms when using a constant number of queries.
    For bent functions, they succeed with probability~1, while for arbitrary functions the success probability depends on the `bentness' of the function.
\end{abstract}

\tableofcontents

\section{Introduction}

Quantum computers are known to be able to solve problems that classical computers are practically incapable of.
For example, the integer factoring and discrete logarithm problems are presumed to be computationally hard on classical computers; this lies at the basis of cryptosystems like RSA \cite{Rivest_1983} and elliptic curve cryptography \cite{Koblitz_1987,miller1985use}.
However, Shor showed in 1994 that both problems are efficiently solvable on a quantum computer \cite{Shor}.
A natural generalization of these problems is the \emph{hidden subgroup problem}.

\begin{Pro}[Hidden subgroup problem]\label{HSP1}
	Consider a finite group $G$, a set $S$ and a function $f: G \rightarrow S$. Suppose there exists a subgroup $H \leq G$ such that $f$ is constant on each coset of~$H$, and distinct on different cosets: $f (g) = f (g')$ if and only if $gH = g'H$. Given oracle access to $f$, find $H$.
\end{Pro}

Other instances of the hidden subgroup problem are Simon's problem \cite{Simon_1997}
and the discrete logarithm problem \cite{Shor_1999}.
\Cref{HSP1} can be solved by a polynomial-time quantum algorithm for finite abelian groups;
more precisely, $H$ can be determined in time $\poly(\log|G|)$ \cite{Childs_2010}.
In the instances that can be efficiently solved, the non-abelian case has applications to the graph automorphism problem,
the graph isomorphism problem and lattice problems \cite{Childs_2010}.
An efficient solution to the non-abelian hidden subgroup problem for dihedral groups would have great impact on lattice cryptosystems \cite{regev2004quantum}.
In contrast to the abelian case, the question whether there is a polynomial-time
quantum algorithm to solve the problem for non-abelian groups remains open,
except in some special cases \cite{Childs_2010}.

A problem related to the hidden subgroup problem is the \emph{hidden shift problem}. 

\begin{Pro}[Hidden shift problem]\label{original hidden shift problem}
	Let $G$ be a finite group, let $S$ be a set and let $f$, $g : G \rightarrow S$ be functions such that there exists a unique $s \in G$ satisfying $g(x) = f(x-s)$ for all $x\in G$.
	Given oracle access to $f$ and~$g$, determine $s$.
\end{Pro}

From now on, we will only consider the case where the group $G$ is abelian, and we write $G$ additively. One approach to solve hidden shift problem is via Fourier sampling methods. The hidden shift problem for~$G$ is equivalent to the hidden subgroup problem for the group $G\rtimes_\varphi \Z/2\Z$, where $\varphi:\Z/2\Z\rightarrow
\Aut(G)$ is given by $\varphi(0):x \mapsto x$ and $\varphi(1):x \mapsto -x$.
In particular, the hidden shift problem for finite cyclic groups $G$ is equivalent to the dihedral hidden subgroup problem \cite{Childs_2010}. The main idea of Fourier sampling methods is to apply a function (e.g. hiding a subgroup) to a uniform superposition, then apply a Fourier transform and measure in the Fourier basis. Sufficient sampling should then give information about the hidden structure of the function. No polynomial-time quantum algorithm is known for the general abelian hidden shift problem.
However, some algorithms are known to solve the problem more efficiently than by brute force, in particular Kuperberg's algorithm and its variants \cite{kuperberg2005subexponential,regev2004subexponential,kuperberg2011another, peikert2020he}, which solve the problem in subexponential time, and another algorithm in \cite{friedl211091hidden}, which solves the problem efficiently for solvable groups with constant exponent and constant derived length. In the case of boolean functions $f: \Z_2^n\to \Z_2$ the query complexity of \cref{original hidden shift problem} depends heavily on the type of functions, specifically on the flatness of their Fourier spectra. The extreme cases are bent functions and delta functions, needing $O(n)$ queries \cite{gavinsky2011quantum} and $\Theta(2^n)$ queries \cite{grover1996fast} respectively. All other boolean functions lie somewhere in between \cite{childs2013easy}. 

An alternative formulation of the hidden shift problem replaces the finite set $S$ in \cref{original hidden shift problem} with (a finite subset of) the complex numbers $\C$.
Previous work has mostly focused on the case of boolean $\pm1$-valued functions \cite{rotteler2009gowersnorm,rotteler2010quantum,gavinsky2011quantum,childs2013easy} or functions valued at roots of unity \cite{van_Dam_2006}.
We relax this restriction and allow the values of $f$ to be arbitrary complex numbers or even arbitrary vectors in $\C^d$.
Following \cite{rotteler2009gowersnorm,rotteler2010quantum}, we assume oracle access to $\hf$, the Fourier transform of $f$, which we use as a reference for determining the hidden shift $s$ of $g$.

\begin{Pro}[Complex-valued hidden shift problem]\label{pro: hidden shift v3}
Let $G$ be a finite abelian group and let $f,g: G \to \C^d$ be functions such that there exists a unique $s\in G$ satisfying $g(x) = f(x-s)$ for all $x\in G$. Given oracle access to $g$ and the Fourier transform of $f$, determine $s$.
\end{Pro}

The standard approach \cite{van_Dam_2006} for solving this version of the problem also starts by applying the oracle for the shifted function $g$ to the uniform superposition and then applying the Fourier transform. Then, instead of measuring, we apply the oracle for $\hf$ to remove all information about $f$ except for the hidden shift $s$, which can be extracted by applying the inverse Fourier transform (see \cref{algo1} for more details). In contrast to the Fourier sampling approach, this approach does not involve intermediate measurements and relies purely on constructive and destructive interference.

A class of functions that is particularly amenable to this approach are the so-called bent functions (see \cref{def: bent function intro}). For boolean bent functions, \cref{pro: hidden shift v3} can be solved efficiently using only a constant number of oracle queries \cite{rotteler2009gowersnorm, rotteler2010quantum}. However, it was recently shown that for the class of Maiorana--McFarland bent functions of bounded degree Rötteler's algorithm also admits an efficient classical simulation \cite{amy2024}.
Earlier, the same ideas as in Rötteler's algorithm had also been applied to multiplicative characters of finite fields and of rings of the form $\Z/N\Z$ \cite{van_Dam_2006}, which are almost bent functions.

In this paper, we extend the above methods for \cref{pro: hidden shift v3} to a larger class of functions. While boolean functions and group characters map to complex roots of unity, we allow for vector-valued functions with image beyond the unit circle or unit sphere. After treating the case of bent functions, we follow the idea in \cite{van_Dam_2006} and look at complex functions that are close to bent (we formalize a notion of approximate `bentness' in the next section). In line with the results on boolean functions \cite{childs2013easy}, the success probability of the algorithms we present increases the more bent a complex function is (for precise statements, see \cref{subsec: our results}).

The quantum algorithms we present show quite some similarities with the algorithm for the forrelation problem \cite{aaronson2015forrelation}. Given two functions $f$ and $g$, the forrelation problems asks to determine whether there is a high or low correlation between $f$ and the Fourier transform of $g$. For bent functions, the forrelation algorithm \cite[Proposition 6]{aaronson2015forrelation} is identical to \cref{algo1} for hidden shifts. In \cref{app: forrelation}, we extend the definition of forrelation to abelian groups.

\subsection{Precise problem statement}\label{sec:problem}

The functions for which we solve the hidden shift problem are generalizations of bent functions.
The definition below generalizes \emph{boolean} bent functions \cite{tokareva2015bent}, for which $G$ is a product of copies of $\Z/2\Z$ and $f$ takes values in $\{\pm1\}$, by making use of the Fourier transform from \cref{def: Fourier tranform}.

\begin{Def}[Bent function]\label{def: bent function intro}
    Let $G$ be a finite abelian group.
    A function $f: G \rightarrow \C$ is called \emph{bent} if
    $|f(x)| = 1$ for all $x \in G$ and $|\hat{f}(\phi)| = 1$ for all $\phi \in \hat{G}$.
\end{Def}

\begin{Rem}\label{rem:bent}
The term `bent function' is usually reserved for boolean functions $f:\set{0,1}^n \to \set{0,1}$ \cite{tokareva2015bent}.
This notion relates to our \cref{def: bent function intro} by letting $G = (\Z/2\Z)^n$ and treating the output of $f$ as $\pm1$-valued via $(-1)^{f(x)}$.
This easily extends to arbitrary prime characteristic $p$ by letting $G = \F_{p^n}$ and treating the output of $f$ as a power of $\omega_p$, a $p$-th root of unity \cite[Chapter~13]{Mesnager}.
The concept we introduce in \cref{def: bent function intro} is more general in that the values of $f$ must be on the unit circle but need not be roots of unity.
The same concept has been introduced independently in different areas and goes by several different names.
Most commonly, the vector $\ket{f} = \sum_{g \in G} f(g) \ket{g}$ associated to a bent function $f$ is called \emph{biunimodular} \cite{fuhr2015,Gabidulin2002}.
Another established term is \emph{constant amplitude zero autocorrelation (CAZAC) sequence} \cite{Benedetto2019}.
\end{Rem}

\begin{Ex}
    For $G = \Z/2\Z$, the function $f$ such that $f(0) =1$ and $f(1)=i$ is bent.
    More generally, the bent functions on $\Z/2\Z$ are exactly the functions $f$ such that $|f(0)|=1$ and $f(1)=\pm if(0)$.
\end{Ex}

The following two definitions relax bentness by only imposing certain bounds on the function and its Fourier transform, either on the whole domain or on a subset thereof.

\begin{Def}[$(R,\hr)$-bounded]\label{def: hr R bounded intro}
Let $G$ be a finite abelian group and let $R, \hr > 0$ be real numbers. A function $f:G \rightarrow \C$ is called \emph{$(R, \hr)$-bounded} if
\[\max_{x\in G} |f(x)| \leq R \quad \text{ and } \quad\min_{\phi \in \hG} |\hf(\phi)| \geq \hr.\]
\end{Def}

\begin{Ex}
For $G = \Z/2\Z$, the function given by $f(0)=1$ and $f(1)=2i$ is $(2,\sqrt{5/2})$-bounded.
\end{Ex}

\begin{Def}[$(r,R,\hr,\hR,\alpha,\halpha)$-bounded]\label{def: r R hr hR bounded intro}
Let $G$ be a finite abelian group and let $0 \leq r \leq R$, $0 \leq \hr \leq \hR$ be real numbers. A function $f:G \rightarrow \C$ is called \emph{$(r,R,\hr,\hR,\alpha,\halpha)$-bounded} if there exist subsets $A \subseteq G, \hA \subseteq \hG$ of sizes $|A| = \alpha |G|, |\hA| = \halpha |\hG|$ such that 
\[|f(x)|\in [r, R] \text{ for all }x\in A, \quad \text{ and } \quad |\hf(\phi)|\in [\hr, \hR]\text{ for all }\phi\in \hA.\]
\end{Def}

\begin{Ex}
All bent functions are $(1,1)$-bounded. All $(R,\hr)$-bounded functions are \ $(0,R,\hr,\infty,\allowbreak 1,1)$-bounded.
\end{Ex}

We can generalize the above definitions even further by looking at multidimensional bent functions and their relaxations.

\begin{Def}\label{def: multidimensional all bounded}
    A multidimensional complex function $f: G \rightarrow \C^d$ is called \emph{bent} if $\|f(x)\| = 1$ for all $x\in G$ and $\|\hf(\phi)\|=1$ for all $\phi \in \hG$ (the Fourier transform of $f$ is determined coordinate-wise). We call $f$ \emph{$(r,R, \hr, \hR, \alpha, \halpha)$-bounded} if there exist subsets $A \subseteq G, \hA \subseteq \hG$ of sizes $|A| = \alpha |G|, |\hA| = \halpha |\hG|$ such that
    \[\|f(x)\|\in [r, R] \text{ for all }x\in A, \quad \text{ and } \quad \|\hf(\phi)\|\in [\hr, \hR]\text{ for all }\phi\in \hA.\]
\end{Def}

With these generalizations of bent functions, we can define the oracles in \cref{pro: hidden shift v3}. Let $f: G\rightarrow\C^d$ be an $(r, R, \hr, \hR, \alpha, \halpha)$-bounded function and consider its shift $g(x) = f(x-s)$ for some fixed unknown $s \in G$. We model oracle access to $g$ by an operator $O_g$ acting by permuting the standard basis vectors of $\C^G \otimes \C^2 \otimes \C^{2^n}$ via
\begin{equation}\label{eq: oracle access g}
    O_g \ket{x} \ket{a} \ket{b} =
    \ket{x}
    \ket{a \oplus \delta_{x-s \in A}}
    \ket{b \oplus g(x)},
\end{equation}
where $x \in G$ in the first register denotes the input to $g$, the second or \emph{indicator} register indicates whether the (shifted) input belongs to $A \subseteq G$, and the third register returns the value of $g$ encoded in $n$ bits (the addition in last two registers is modulo two).
Note from \cref{def: multidimensional all bounded} that $\delta_{x-s\in A} = 1$ if and only if $\norm{g(x)} \in [r,R]$, so the indicator register can be post-selected to ensure that $g$ is within these bounds.
We will assume throughout that all function values $g(x) \in \C^d$ can be stored in an $n$-bit register with complete precision. In practice, these complex vectors (or numbers) may need to be approximated to fit within $n$ bits (see \cref{sec:error-analysis} for more details).

Quantum access to $\hf$ is similarly modeled by an oracle $O_{\hf}$ acting on $\C^{\hG} \otimes \C^2 \otimes \C^{2^n}$ via
\begin{equation}\label{eq: oracle access hf}
    O_{\hf} \ket{\phi} \ket{a} \ket{b} =
    \ket{\phi}
    \ket{a \oplus \delta_{\phi \in \hA}}
    \ket{b \oplus \hf(\phi)}.
\end{equation}
Note from \cref{def: multidimensional all bounded} that $\delta_{\phi \in \hA} = 1$ if and only if $\norm{\hf(\phi)} \in [\hr,\hR]$.
Both oracles are linear extensions of bijections, meaning they can also be used classically, allowing for a fair comparison.
Since addition is modulo two, both oracles are self-inverse: $O_g^2 = O_{\hf}^2 = I$.

\begin{Rem}
    For an $(r,R, \hr,\hR, \alpha, \halpha)$-bounded function with $\alpha=1$ (resp.\ $\halpha=1$) we ignore the indicator register for $O_g$ (resp.\ $O_{\hf}$), as the indicator is then the constant function $1$.
\end{Rem}

Using these definitions we state a precise version of \cref{pro: hidden shift v3} that we will solve.

\begin{Pro}[Complex-valued hidden shift problem, precise statement]\label{problem: exact formulation}
    Let $G$ be a finite abelian group and consider an $(r, R, \hr, \hR, \alpha, \halpha)$-bounded function $f: G \rightarrow \C^d$. Consider a function $g: G\rightarrow \C^d$ such that there is a unique $s\in G$ satisfying $g(x) = f(x-s)$ for all $x\in G$. Given access to the quantum oracles $O_g$ and $O_{\hf}$, determine $s$.
\end{Pro}

The above statement captures the most general problem that we solve in this paper.
For the purpose of presentation, we first consider special cases where some of the parameters are dropped to make the problem simpler.
We summarize our results in the next section.

\subsection{Our results}\label{subsec: our results}

We present a suite of increasingly general quantum algorithms for the hidden shift problem, the most general of which applies to functions $f: G\rightarrow \C^d$ for any abelian group $G$ and dimension $d \geq 1$. Our algorithms are adapted from \cite[Section~3]{van_Dam_2006} and they find $s$ using only four oracle queries.
For bent functions, the algorithms are exact, but for other classes of functions their success probability depends on the `bentness' of the function $f$. In line with previous results on the hidden shift problem for boolean functions \cite{childs2013easy}, our algorithms become less effective when the function is further from bent and closer to a delta function.

The paper is organized as follows. In \cref{sec:preliminaries}, we recall the notions of quantum circuits and gate complexity, give some basic results on characters of finite abelian groups and define the Fourier transform on a finite abelian group. The next four sections contain algorithms that can be applied to increasing generalizations of bent functions, and their success probabilities are analyzed. Specifically, we obtain the following results.

\begin{itemize}
    \item In \cref{sec:exact} we present a classical and quantum algorithm (\cref{alg: exact classical,algo1}) to solve the hidden shift problem for one-dimensional bent functions with probability one (\cref{thm: exact classical,thm:bent}). 
    \item In \cref{subsec: first gen} we prove \cref{thm: first gen}: Let $f: G\rightarrow \C$ be an $(R, \hr)$-bounded function. Then there exists a quantum algorithm (\cref{alg: approx1}) using four queries that finds the hidden shift $s$ with probability
    \begin{equation}\label{eq:1}
    p(s) = \left(\frac{\hr}{R}\right)^2.
    \end{equation}
    When $f$ is a bent function, we have $\hr=R=1$ and we find the shift with certainty.
    \item In \cref{subsec: second gen} we prove \cref{thm: second gen}: Let $f: G \rightarrow \C$ be an $(r, R, \hr, \hR, \alpha, \halpha)$-bounded function. Then there exists a quantum algorithm (\cref{alg: approx2}) using four queries that finds the hidden shift $s$ with probability
    \begin{equation}\label{eq:2}
        p(s) = \left(\frac{\hr}{R}\right)^2\left|\halpha - \frac{1}{|G|^{3/2}}\sum_{\phi \in \smash\hA}\sum_{x \notin A+s} \phi(x)\overline{\phi(s)} \frac{g(x)}{\hf(\phi)}\right|^2.
    \end{equation}
    By taking $A=G$, $\hA=\hG$, $\hR=\infty$ and $r=0$ this can be shown to reduce to \cref{eq:1}. Not all parameters in the set $(r, R, \hr, \hR, \alpha, \halpha)$ can directly be found in the success probability. We elaborate in \cref{subsec: symmetry} why they are still included in the definition.
    \item In \cref{sec: multi} we look at higher-dimensional bent functions and prove \cref{thm: multidim gen}: Let $f = (f_0, \dots, f_{d-1}): G \rightarrow \C^d$ be a $d$-dimensional $(r, R, \hr, \hR, \alpha, \halpha)$-bounded function. Then there exists a quantum algorithm (\cref{alg: generalised bent multid}) using four queries that finds the hidden shift $s$ with probability
    \[
        p(s) = \left(\frac{\hr}{R}\right)^2\left|\halpha - \frac{1}{|G|^{3/2}}\sum_{\phi \in \smash\hA}\sum_{x \notin A+s} \phi(x)\overline{\phi(s)} \sum_{i=0}^{d-1}\frac{g_i(x)\overline{\hf_i(\phi)}}{\|\hf(\phi)\|^2}\right|^2.
    \]
    It reduces to \cref{eq:2} by taking $d=1$ and noting that in one dimension $\frac{\overline{\hf(\phi)}}{\|\hf(\phi)\|^2} = \frac{1}{\hf(\phi)}.$
\end{itemize}
The above generalizations are presented in several steps to highlight the different parts of the most basic algorithm (\cref{algo1}) that need to be adapted to solve the most general version of the problem (\cref{problem: exact formulation}).

Since $\hr \leq R$ due to Parseval's identity \eqref{eq:Parseval}, the probability in \cref{eq:1} equals one only when $\hr = R = 1$, meaning that the function is bent.
In \cref{sec: one register} we investigate whether the success probability for non-bent functions can be increased by allowing a more powerful oracle that contains additional tunable phase degrees of freedom $\theta$ and $\chi$.
We show that these phases lead to additional interference when utilizing only one instead of two ancillary qubits, and prove in \cref{Thm:approximate_algorithm} that this modified algorithm achieves success probability
\[
p(s)
        = \abs*{
            \frac{\hr}{R}
          + \frac{1}{|G|^{3/2}}
            \sum_{x \in G}
            \sum_{\phi \in \Ghat}
            \phi(x)
            e^{i\theta(x+s)}
            \sqrt{1-\abs*{\frac{f(x)}{R}}^2}
            e^{-i\chi(\phi)}
            \sqrt{1-\abs*{\frac{\hr}{\hat{f}(\phi)}}^2}
        }^2.
\]
By appropriately tuning $\theta$ and $\chi$, we can obtain success probability one even for some non-bent functions (\cref{thm:prob-all1}).

\section{Preliminaries}\label{sec:preliminaries}

\subsection{Quantum circuits and complexity}\label{sec:qcircuits}

We implicitly use the quantum circuit model (see \cite[Chapter 4]{nielsen2000quantum}) which describes quantum algorithms by sequences of quantum gates (such as single-qubit gates and CNOT) and measurements.
More precisely, a quantum algorithm for a computational problem is a family of circuits where the number of qubits depends only on the input size.

We do not consider parallel gates, and define the time complexity of a quantum algorithm as the number of gates in the corresponding quantum circuit.
Viewed as a function of the input size, time complexity is often expressed in big-$O$ notation.
A quantum algorithm is called \emph{efficient} if the number of gates is polynomial in the input size.

As mentioned in the introduction, we need to store the elements of a finite subset of $\C^d$ in an $n$-qubit register, and we assume this can be done with complete precision.
We assume the encoding has been chosen such that certain unitary operators appearing in our algorithms can be implemented efficiently, in particular the operators $S$, $U_1$, $V_{\rot}$ and $U_2$ defined in \cref{eq:S,eq:U1,eq:Vrot,eq:U2}.
Furthermore, when considering a finite abelian group $G$, we assume that both the elements of $G$ and of its character group $\hat{G}$ are encoded in quantum registers in such a way that the quantum Fourier transform operator $F$ defined in \cref{eq:F} is efficiently implementable. For example, it suffices to describe $G$ as a product of cyclic groups, so we can use the implementation from \cite{kitaev1995quantummeasurementsabelianstabilizer}.

In the remainder of this paper, we will not explicitly mention time complexity any further, but under the above assumptions all the quantum algorithms we give are efficient in the above sense.

\subsection{Group characters and Fourier transform}\label{subsec: group and Fourier transform}

Throughout the paper, $G$ is a finite abelian group and $f: G \rightarrow \C$ is a complex function. We denote by $|G|$ the order of $G$. A \emph{character} of $G$ is a map $\phi: G \rightarrow \C^\times$ such that $\phi(x + y) = \phi(x) \phi(y)$ for all $x,y \in G$, i.e., $\phi$ is a homomorphism.
The characters of~$G$ form a group under point-wise multiplication: the product of $\phi_1$ and $\phi_2$ is the character $\phi$ given by
$\phi(x) = \phi_1(x) \phi_2(x)$ for all $x \in G$.
This group is called the \emph{character group} of $G$ and denoted by~$\hat{G}$.
It is well known that $G$ and $\hat{G}$ are (non-canonically) isomorphic.

If $G = \Z/N\Z$ for a positive integer $N$, then the map $\hat{G} \rightarrow \C^\times$ defined by $\phi\mapsto\phi(1)$ is a group isomorphism from $\hat{G}$ to the group of $N$-th roots of unity, with inverse given by $\zeta \mapsto (a \mapsto \zeta^a)$.
More generally, if $G = \Z/N_1\Z\times \dots \times \Z/N_l\Z$ for positive integers $N_j$, then $\hat{G}$ is generated by $\phi_1,\dots,\phi_l$, where
\begin{equation}
    \label{eq:phi}
    \phi_j(x_1,\dots, x_l) = \exp(2\pi i x_j/N_j)
    \quad\text{for } 1\leq j \leq l.
\end{equation}

Taking $0$ and $\phi_0$ as the zero elements of the groups $G$ and $\hat{G}$ respectively, we have
\begin{align*}
    \sum_{x\in G} \phi(x) &=
    \begin{cases}
        |G| & \text{if } \phi = \phi_0, \\
        0 & \text{otherwise,}
    \end{cases} &
    \sum_{\phi\in\hat{G}} \phi(x) &=
    \begin{cases}
        |G| & \text{if } x = 0,\\
        0 & \text{otherwise}
    \end{cases}
\end{align*}
thanks to orthogonality of characters.

A central tool in our quantum algorithms is the Fourier transform for abelian groups.

\begin{Def}[Fourier transform]\label{def: Fourier tranform}
	For any function $f: G \rightarrow \C$, the \emph{Fourier transform} of $f$ is the function
    \begin{align*}
        \hat{f}: \hat{G} &\rightarrow \C\\
        \phi&\mapsto \frac{1}{|G|^{1/2}}\sum_{x\in G}\phi(x)f(x).
    \end{align*}
\end{Def}

The following lemma describes a key property of the Fourier transform: it translates shifts into point-wise multiplication.
Since $g(x) = f(x-s)$ is the convolution of $f$ with the delta function at~$s$, this is a special case of the fact that the Fourier transform translates convolution into point-wise multiplication.

\begin{Lem}\label{lem:ghat}
	Let $f: G \rightarrow \C$ be a function, let $s\in G$, and let $g(x) = f(x-s)$. Then
	\[
        \hat{g}(\phi) = \phi(s) \hat{f}(\phi) \quad \text{for all } \phi \in \hat{G}.
	\]
\end{Lem}

\begin{proof}
	This follows by computing the Fourier transform of $g$, using $\phi(x) = \phi(s) \phi(x-s)$, and making the change of variables $y = x - s$:
	\begin{align*}
    	\hat{g}(\phi) & = \frac{1}{|G|^{1/2}}\sum_{x\in G}\phi(x)g(x)\\
    	& = \frac{1}{|G|^{1/2}}\sum_{x\in G}\phi(x)f(x-s)\\
    	& = \frac{1}{|G|^{1/2}}\sum_{x\in G}\phi(s)\phi(x-s)f(x-s)\\
    	& = \phi(s)\frac{1}{|G|^{1/2}}\sum_{y\in G}\phi(y)f(y)\\
    	& = \phi(s)\hat{f}(\phi).
        \qedhere
	\end{align*}
\end{proof}

\begin{Lem}\label{lem:shift-bent}
    If $f: G \rightarrow \C$ is bent, then so is $g(x) = f(x-s)$ for any $s \in G$.
\end{Lem}

\begin{proof}
    Clearly, we have $|g(x)|=1$ for all $x\in G$. Using \cref{lem:ghat} and the fact that $|\phi(s)| =1$, we obtain $|\hat g(\phi)|=1$ for all $\phi\in \hat{G}$.
\end{proof}

\subsection{The quantum Fourier transform operator}

Let $\CG$ and $\CGhat$ be finite-dimensional complex vector spaces with bases indexed by $G$ and $\hat{G}$, respectively, and equipped with the standard hermitian inner product for which these bases are orthonormal. Formally,
\[
    \CG = \left\{\sum_{x\in G}\alpha_x\ket{x}:\alpha_x\in \C \text{ for all }x\in G\right\},
\]
and similarly for $\CGhat$.

The \emph{Fourier transform} on $G$ is the unitary operator $F: \CG\rightarrow\CGhat$ defined as
\begin{equation}
    F = \frac{1}{|G|^{1/2}}\sum_{\phi\in\smash{\hat{G}}}\sum_{x\in G}\phi(x)\kb{\phi}{x},
    \label{eq:F}
\end{equation}
and its inverse is the unitary operator $F^\dagger: \CGhat\rightarrow\CG$
defined as
\[
    F^\dagger = \frac{1}{|G|^{1/2}}\sum_{x\in G}\sum_{\phi\in\smash{\hat{G}}}\overline{\phi(x)}\kb{x}{\phi}.
\]
The operators $F$ and $F\ct$ act on the basis vectors of their respective input spaces as follows:
\begin{equation}\label{eq:F on basis vectors}
    F\ket{x} = \frac{1}{|G|^{1/2}}\sum_{\phi\in \hat{G}}\phi(x)\ket{\phi},
    \quad \text{and} \quad
    F^{\dagger}\ket{\phi} = \frac{1}{|G|^{1/2}}\sum_{x\in G}\overline{\phi(x)}\ket{x}.
\end{equation}
In particular, $F$ maps any superposition with amplitudes $f$ to one with amplitudes $\hf$:
\begin{equation}\label{eq:f to hf}
    F \cdot \frac{1}{|G|^{1/2}} \sum_{x \in G} f(x) \ket{x}
    = \frac{1}{|G|^{1/2}} \sum_{x \in G} \frac{1}{|G|^{1/2}} \sum_{\phi\in \smash{\hat{G}}} \phi(x) f(x) \ket{\phi}
    = \frac{1}{|G|^{1/2}} \sum_{\phi \in \hG} \hf(\phi) \ket{\phi}.
\end{equation}
Thanks to unitarity of $F$ we get \emph{Parseval's identity}
\begin{equation}\label{eq:Parseval}
    \sum_{x \in G} |f(x)|^2
    = \sum_{\phi \in \Ghat} |\hf(\phi)|^2.
\end{equation}

\section{Exact algorithm for bent functions}\label{sec:exact}

Recall from \cref{def: bent function intro} that a function $f: G \rightarrow \C$ is bent if $|f(x)| = |\hat{f}(\phi)| = 1$ for all $x \in G$ and $\phi \in \hat{G}$.
In this section, we present classical and quantum algorithms that solve the hidden shift \cref{problem: exact formulation} with $d=1$ for the class of bent functions exactly, i.e.\ with success probability~1.

\subsection{Classical algorithm}

We present a simple classical algorithm for \cref{problem: exact formulation}, based on \cref{lem:ghat}.
We do not claim that it has optimal time or query complexity, but we state it for the sake of having something to compare our quantum algorithms against.

Let $G$ be a finite abelian group given as a product
\[G = \prod_{j=1}^{l} \Z/N_j\Z\]
of cyclic groups.
The dual group $\hat{G}$ is then generated by characters $\phi_j$ as described in \cref{eq:phi}.
For an unknown bent function $f: G \rightarrow \C$, given oracle access to a shifted function $g: G \rightarrow \C$ and $\hf: \hG \rightarrow \C$, the goal is to determine the unique hidden shift $s = (s_1, \dots, s_l)$.

\begin{Algo}\label{alg: exact classical}\
\begin{enumerate}
\item Query $g$ on all elements of $G$ and store the values.
\item For $j \in \{1,\dots, l\}$, repeat the following steps:
\begin{enumerate}
    \item compute $\hat{g}(\phi_j)=\frac{1}{\sqrt{|G|}}\sum_{x\in G}\exp(2\pi i x_j/N_j)g(x)$ using the stored values of~$g$,
    \item query $\hat{f}(\phi_j)$ and compute $\phi_j(s) = \frac{\hat{g}(\phi_j)}{\hat{f}(\phi_j)}$,
    \item compute $s_j = \frac{N_j\log(\phi_j(s))}{2\pi i}\bmod N_j$.
\end{enumerate}
\item Return $s = (s_1, \dots, s_l)$.
\end{enumerate}
\end{Algo}
\begin{Thm}\label{thm: exact classical}
    \Cref{alg: exact classical} finds the hidden shift $s$ with certainty using $|G|$ queries to~$O_g$, $l$ queries to~$O_{\hf}$ and $O(l|G|)$ arithmetic operations in~$\C$ (including complex exponentials and logarithms).
\end{Thm}
\begin{proof}
The first step uses $|G|$ queries to $O_g$.
The $l$ iterations of step~2(a) use $O(l|G|)$ arithmetic operations in~$\C$.  (The fast Fourier transform is of no use here since we only need to compute $l$ values of $\hat g$.)
In the $l$ iterations of step~2(b), we indeed obtain the correct values of $\phi_j(s)$ by \cref{lem:ghat}, using $l$ queries to $O_{\hf}$ and $l$ arithmetic operations.
Finally, we obtain the correct $s_1,\dotsc,s_l$ in step~2(c) because of~\cref{eq:phi}, using $O(l)$ arithmetic operations in~$\C$.
\end{proof}

\subsection{Quantum algorithm}\label{sec:exact q alg}

For the quantum algorithm we assume the same setting of a finite abelian group $G$ and two bent functions $f,g: G \rightarrow \C$ such that $g(x) = f(x-s)$ for some unique shift $s$.
We assume access to the quantum oracles $O_g$ and $O_{\hf}$ as given in \cref{eq: oracle access g} and \cref{eq: oracle access hf}, where we ignore the indicator register.
Because the images of $g$ and $\smash{\hf}$ are on the unit circle, we can define new quantum oracles $\tilde{O}_{g}\in \U{\C^G}$ and $\tilde{O}_{1/\hf} \in \U{\C^{\hG}}$ by
\begin{equation}\label{eq: phase oracles}
    \tilde{O}_g\ket{x} = g(x)\ket{x}, \quad    \tilde{O}_{1/\hf}\ket{\phi} = \hf(\phi)^{-1}\ket{\phi}.
\end{equation}
We also assume that we have access to an operator $S$ acting on a quantum register containing complex numbers $z$ of absolute value 1 by
\begin{equation}
S\ket{z}=z\ket{z}.
\label{eq:S}
\end{equation}
In practice, $z$ will have finite precision, so we can implement $S$ using a sequence of controlled rotation operators.
This may introduce numerical errors, the effects of which are discussed in \cref{sec:error-analysis}.

\begin{Lem}\label{lem: phase oracle 1-dim}
    The quantum oracles $\tilde{O}_g$ and $\tilde{O}_{1/\hat{f}}$ in \cref{eq: phase oracles} can be efficiently implemented using two queries to $O_g$ and to $O_{\hf}$, respectively.
\end{Lem}

\begin{proof}
    Since $|g(x)| = 1$ for all $x\in G$, we can implement $\tilde{O}_g$ using $O_g$ and the operator $S$ defined in~\cref{eq:S} as
    \begin{align*}
        \ket{x}\ket{0}&\xmapsto{O_g} \ket{x}\ket{g(x)}\\
        &\xmapsto{I\otimes S} g(x)\ket{x}\ket{g(x)}\\
        &\xmapsto{O_g}g(x)\ket{x}\ket{0}.
    \end{align*}
    One can similarly implement $\tilde{O}_{1/\hf}$ using $O_f$ and the inverse of~$S$.
\end{proof}

\begin{Rem}
Under somewhat stronger assumptions on $g$ and $O_g$, a single call to $O_g$ suffices to implement $\tilde{O}_g$.
Namely, let us assume that $g$ takes values in a known finite subgroup $C$ of $\U{\C}=\{z\in\C:|z|=1\}$.  Then the dual group $\hat C$ is canonically isomorphic to $\Z/|C|\Z$, and the quantum Fourier transform $F_C$ sends the uniform superposition $|C|^{-1/2}\sum_{z\in C}\ket{z}$ to $\ket{0}$, where $0$ denotes the trivial element of $\hat C$. 
Furthermore, suppose that we have access to a modified version of $O_g$ that takes a state of the form $\ket{x}\ket{z}$ with $x\in G$ and $z\in C$ and sends it to $\ket{x}\ket{zg(x)}$.
Then we can implement $\tilde O_g$ using the following variant of phase kickback (cf.\ \cite[proof of Theorem 4.1]{rotteler2010quantum}):
\begin{align*}
\ket{x}\ket{0}&\xmapsto{I\otimes F_C^\dagger}\frac{1}{|C|^{1/2}}\sum_{z\in C}\ket{x}\ket{z}\\
&\xmapsto{I\otimes S^{-1}}\frac{1}{|C|^{1/2}}\sum_{z\in C}z^{-1}\ket{x}\ket{z}\\
&\xmapsto{O_g}
\frac{1}{|C|^{1/2}}\sum_{z\in C}z^{-1}\ket{x}\ket{zg(x)}\\
&\qquad=\frac{1}{|C|^{1/2}}\sum_{w\in C}(wg(x)^{-1})^{-1}\ket{x}\ket{w}=\frac{g(x)}{|C|^{1/2}}\sum_{w\in C}w^{-1}\ket{x}\ket{w}\\
&\xmapsto{I\otimes S}\frac{g(x)}{|C|^{1/2}}\sum_{w\in C}\ket{x}\ket{w}\\
&\xmapsto{I\otimes F_C} g(x)\ket{x}\ket{0}.
\end{align*}
Similarly, we can implement $\tilde O_{1/\hf}$ using a single call to a modified oracle for $\hf$ under analogous assumptions.
\end{Rem}

Below we describe an exact and efficient quantum algorithm for \cref{problem: exact formulation}, which uses the oracles from \cref{eq: phase oracles} and generalizes an earlier algorithm for \emph{boolean} bent functions \cite[Theorem~4.1]{rotteler2010quantum}. Our algorithm uses the state space $\C^G$.

\begin{Algo}~\label{algo1}
	\begin{enumerate}
		\item Prepare the uniform superposition
		$|G|^{-1/2} \sum_{x\in G} \ket{x}$.
		\item Apply the operator $F^\dagger \tilde{O}_{1/\hat{f}} F \tilde{O}_g$.
		\item Measure in the standard basis.
	\end{enumerate}
\end{Algo}

\begin{Thm}\label{thm:bent}
	If $f$ is a bent function, then \cref{algo1} determines the hidden shift $s$ with certainty using one call to the oracle $\tilde{O}_g$ and one call to the oracle $\tilde{O}_{1/\hat{f}}$.
\end{Thm}

\begin{proof}
The state changes as
\begin{align*}
    \frac{1}{|G|^{1/2}}
    \sum_{x\in G} \ket{x}
    &\xmapsto{\tilde{O}_g}
      \frac{1}{|G|^{1/2}}\sum_{x\in G}g(x)\ket{x}
    	 \xmapsto{F}
      \frac{1}{|G|^{1/2}}\sum_{\phi\in\hat{G}}\hat{g}(\phi)\ket{\phi} =
      \frac{1}{|G|^{1/2}}\sum_{\phi\in\hat{G}}\phi(s)\hat{f}(\phi)\ket{\phi} \\
    &\xmapsto{\tilde{O}_{1/\hat{f}}}
      \frac{1}{|G|^{1/2}}\sum_{\phi\in\hat{G}}\phi(s)\ket{\phi}
      \xmapsto{F^{\dagger}}
      \ket{s},
\end{align*}
where we made use of \cref{eq:f to hf,lem:ghat,eq:F on basis vectors}.
\end{proof}

All our subsequent algorithms rely on the same principle as \cref{algo1}. Namely, the state remains in a uniform superposition at key steps throughout the algorithm, with all information encoded in relative phases between standard or Fourier basis states. We manipulate the phase function with oracle calls and Fourier transform in the following way:
\begin{equation} 1 \xmapsto{\tilde{O}_g} g \xmapsto{F}\hg = \hat{\delta}_s \cdot \hf \xmapsto{\tilde{O}_{1/\hf}} \hat{\delta}_s \xmapsto{F^{\dagger}}\delta_s,
\end{equation}
allowing us to learn $s$ at the final step. Although in subsequent sections the algorithms and calculations become more difficult, the underlying idea stays the same.

\section{Approximate algorithms for bent-like functions}\label{sec:approximate}

The hidden shift problem for bent functions can be solved with certainty by \cref{algo1} using one query to each of the two phase oracles $\tilde{O}_g$ and $\tilde{O}_{1/\hf}$ from \cref{eq: phase oracles}. If we want to extend the algorithm to functions $f: G \rightarrow \C$ that are not bent, we cannot use phase oracles anymore.
Recall from \cref{lem: phase oracle 1-dim} that each phase oracle $\tilde{O}_g$ and $\tilde{O}_{1/\hf}$ can be implemented by two calls to one of the original oracles $O_g$ and $O_{1/\hf}$ from \cref{eq: oracle access g,eq: oracle access hf}, respectively.
In all subsequent extensions of \cref{algo1} we rely on these two original oracles.

The main idea is to shift the encoding of $f$ from phase to the whole amplitude.
Since $G$ is finite, there exists an $R > 0$ such that $|f(x)|\leq R$ for all $x\in G$, so the function $f(\cdot)/R$ maps to the unit disc and hence
\begin{equation}\label{eq: phase explanation}
    \underbrace{\vphantom{\sqrt{1-\left|\frac{f(x)}{R}\right|^2}\ket{1}}\frac{f(x)}{R} \ket{0}}_{\text{good}} + \underbrace{\sqrt{1-\left|\frac{f(x)}{R}\right|^2}\ket{1}}_{\text{bad}}
\end{equation}
is a valid qubit state. If $|f|$ is valued very close to $R$, most of the amplitude is at $\ket{0}$. The more amplitude ends up in the `good part' the better the algorithm performs. In particular, for bent functions we can take $R=1$ and all of the amplitude ends up in the good part.

In this section, all algorithms require two extra qubits, one for $O_g$ and one for $O_{1/\hf}$, to encode the oracle output into amplitudes according to \cref{eq: phase explanation}.

\subsection{First generalization: \texorpdfstring{$(R,\hr)$-}{(R, hat{r})-}bounded functions}\label{subsec: first gen}

Any complex function on a finite abelian group is bounded below and above. In the first generalization of \cref{algo1}, we additionally need the Fourier transform to vanish nowhere. In other words, we assume our function $f$ to be $(R, \hr)$-bounded as in \cref{def: hr R bounded intro}.

\begin{Prop}
    \label{prop:bounded}
    For all $(R,\hr)$-bounded functions $f$,
    \begin{enumerate}
        \item\label{item0} $\|f\|_\infty \leq R$ and $\|1/\hat{f}\|_\infty \leq 1/\hr$,
        \item\label{item1} $\hr\leq R$, and if $\hr=R$ then $|f(x)|=R$ and $|\hat f(\phi)|=\hr$ for all $x\in G$ and $\phi\in\Ghat$,
        \item\label{item2} for any $s \in G$ the function $g(x) = f(x-s)$ is also $(R,\hr)$-bounded.
    \end{enumerate}
\end{Prop}

\begin{proof}
    Claim~\ref{item0} follows immediately from \cref{def: hr R bounded intro}.
    For claim~\ref{item1}, note that Parseval's identity~\eqref{eq:Parseval} implies
    \begin{equation*}
        |G| \hr^2 \le \sum_{\phi \in \hat{G}}|\hat f(\phi)|^2 =
        \sum_{x \in G}|f(x)|^2 \le |G| R^2.
    \end{equation*}
    Hence $\hr \leq R$, and if $\hr=R$ then both inequalities are equalities, which together with \cref{def: hr R bounded intro} implies the claim.
    For claim~\ref{item2}, note that if $f$ is bounded from above by $R$ then so is $g$. By \cref{lem:ghat}, $|g(\phi)| \geq \hr$ for all $\phi$, hence $g$ is also $(R,\hr)$-bounded.
\end{proof}

Below we give a quantum algorithm for the hidden shift problem of $(R,\hr)$-bounded functions.
We assume access to oracles for $g$ and $\hf$, and use additional unitary operators $U_1$ and $U_2$ that create states of the form \eqref{eq: phase explanation}.
Specifically, we use the following quantum operators:
\begin{itemize}
    \item Quantum oracles $O_g$ and $O_{\hf}$ to access $g$ and $\hf$ given by \cref{eq: oracle access g,eq: oracle access hf} (we ignore the second register as in \cref{sec:exact q alg}).
    \item The quantum operation $U_1 \in \U{\C^{2^n} \otimes \C^2}$ acting as
    $U_1 = \sum_w \proj{w} \x U_1(w)$
    where
	\begin{equation}\label{eq: U1 second gen}
        U_1(w)\ket{0} = \frac{w}{R}\ket{0} + \sqrt{1-|w/R|^2}\ket{1}
	\end{equation}
    and $w$ ranges over some subset of $\C$ that contains both $\set{g(x) : x \in G}$ and $\set{\hf(\phi) : \phi \in \hG}$, and $\ket{w}$ denotes an $n$-bit encoding of $w$ (we do not specify an explicit encoding).
    \item The quantum operation $U_2 \in \U{\C^{2^n} \otimes \C^2}$ acting as
    $U_2 = \sum_w \proj{w} \x U_2(w)$
    where\footnote{We only need to know $U_2(w)\ct \ket{0}$ or $\bra{0} U_2(w)$ since we will analyze the second half of the circuit backwards.}
	\begin{equation}\label{eq: U2 second gen}
        U_2(w)\ct \ket{0} =
        \begin{cases}
            \frac{\hr}{\bar{w}} \ket{0} + \sqrt{1-|\hr/\bar{w}|^2} \ket{1} & \text{if $w \neq 0$}, \\
            \ket{0} & \text{if $w = 0$}.
        \end{cases}
    \end{equation}
	\item The quantum Fourier transform $F \in \U{\C^{G}}$ for $G$, see \cref{eq:F}.
\end{itemize}

The algorithm works on the space $\C^G \otimes \C^{2^n} \otimes \C^2 \otimes \C^2$, where the second register is extra $n$-qubit workspace that is returned to its original state.

\begin{Algo}~\label{alg: approx1}
\begin{enumerate}
    \item Prepare the uniform superposition $\ket{G,0,0,0} = \frac{1}{|G|^{1/2}}\sum_{x\in G}\ket{x,0,0,0}$.
    \item Perform the following quantum circuit:
    
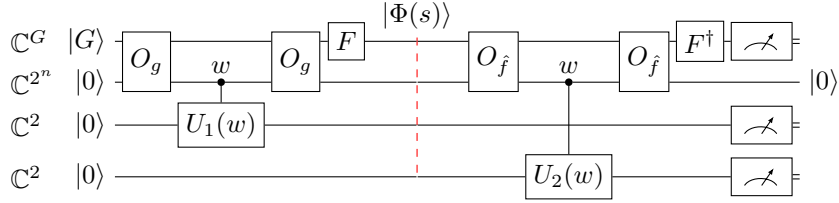
\begin{figure}[H]
    \centering
    \begin{tikzpicture}
        \def\dx{1.5cm} 
    	\begin{yquant}
            [name = regG] qubit {$\ket{G}$} G;
            [name = regw] qubit {$\ket{0}$} w;
            [name = rega] qubit {$\ket{0}$} a;
            [name = regb] qubit {$\ket{0}$} b;
            \path (regG.east)+(-\dx,0) node[anchor=west] {$\C^G$};
            \path (regw.east)+(-\dx,0) node[anchor=west] {$\C^{2^n}$};
            \path (rega.east)+(-\dx,0) node[anchor=west] {$\C^2$};
            \path (regb.east)+(-\dx,0) node[anchor=west] {$\C^2$};

    		  box {$O_g$} (G,w);
            [name = U1] box {$U_1(w)$} a | w;
            \path (U1-p) node[above] {$w$};
            box {$O_g$} (G,w);
    		  box {$F$} G;
              
            [red, label = $\ket{\Phi(s)}$] barrier (-);
            box {$O_{\hf}$} (G,w);
            [name = U2] box {$U_2(w)$} b | w;
            \path (U2-p) node[above] {$w$};
            box {$O_{\hf}$} (G,w);
    		box {$F^{\dagger}$} G;
            
            align G,a,b;
            measure G;
    		measure b;
            measure a;

            output {$\ket{0}$} w;
    	\end{yquant}
    \end{tikzpicture}
    \end{figure}
    \item If the measurements on the third and fourth register result in $0$, output the value of the first register. Otherwise, output \textnormal{FAIL}.
\end{enumerate}
\end{Algo}

\begin{Thm}\label{thm: first gen}
    For an $(R, \hr)$-bounded function $f: G\to \C$, the probability of finding the hidden shift $s$ by \cref{alg: approx1} is
    \begin{equation}\label{eq:basic p(s)}
        p(s) = \left(\frac{\hr}{R}\right)^2.
    \end{equation}
    The algorithm uses two calls to each of the oracles $O_g$ and $O_{\hf}$.
\end{Thm}

\begin{proof}
Observe that the second register is always restored to $\ket{0}$ since each oracle is called twice and the state is not affected by controls. Hence, to successfully find the hidden shift $s$, we would like the final state before the measurement to be $\ket{s,0,0,0}$.
Let $\ket{\Phi(s)}$ denote the state at the marked time point, and similarly let $\ket{\Psi(s)}$ denote the state at the same time point but when running the circuit backwards from the desired target state $\ket{s,0,0,0}$. Then the probability of finding the hidden shift $s$ is
\[p(s) = \left|\braket{\Psi(s)}{\Phi(s)} \right|^2.\]

To compute this probability, we first analyze how the state changes when running the algorithm forwards till the marked time point. Using \cref{eq: oracle access g,eq: U1 second gen},
\begin{align*}
    \frac{1}{|G|^{1/2}} \sum_{x \in G} \ket{x,0,0,0}
    & \xmapsto{O_g}
    \frac{1}{|G|^{1/2}} \sum_{x \in G} \ket{x,g(x),0,0} \\
    & \xmapsto{U_1}
    \frac{1}{|G|^{1/2}} \sum_{x \in G} \ket{x} \ket{g(x)}
    \of*{\frac{g(x)}{R} \ket{0} + \sqrt{1-\abs*{\frac{g(x)}{R}}^2} \ket{1}} \ket{0} \\
    & \xmapsto{O_g}
    \frac{1}{|G|^{1/2}} \sum_{x \in G} \ket{x} \ket{0}
    \of*{\frac{g(x)}{R} \ket{0} + \sqrt{1-\abs*{\frac{g(x)}{R}}^2} \ket{1}} \ket{0}.
\end{align*}
According to \cref{eq:f to hf,eq:F on basis vectors}, the Fourier transform produces
\begin{equation*}
    \ket{\Phi(s)}
    = \frac{1}{|G|^{1/2}} \sum_{\phi \in \hG}
    \ket{\phi} \ket{0}
    \of*{
        \frac{\hg(\phi)}{R} \ket{0}
        +
        \frac{1}{|G|^{1/2}}
        \sum_{x \in G} \phi(x)
        \sqrt{1-\abs*{\frac{g(x)}{R}}^2} \ket{1}
    }
    \ket{0}.
\end{equation*}
Similarly, running the algorithm backwards from $\ket{s,0,0,0}$ to the same time point we get
\begin{align*}
    \ket{\Psi(s)}
    &= O_{\hf} U_2\ct O_{\hf} F \ket{s,0,0,0} \\
    &= \frac{1}{|G|^{1/2}} \sum_{\phi \in \hG}
    \phi(s) \ket{\phi} \ket{0} \ket{0}
    \of*{
        \frac{\hr}{\overline{\hf(\phi)}} \ket{0}
        + \sqrt{1-\abs*{\frac{\hr}{\overline{\hf(\phi)}}}^2} \ket{1}
    },
\end{align*}
where we used
\cref{eq: U2 second gen}.
Putting everything together, the probability of recovering $s$ is
\[p(s) = \abs{\braket{\Psi(s)}{\Phi(s)}}^2 = \left|\frac{1}{|G|}\sum_{\phi\in \hG} \overline{\phi(s)} \frac{\hg(\phi)}{R}\frac{\hr}{\hf(\phi)}\right|^2=\left|\frac{1}{|G|}\sum_{\phi \in \hG} \phi(s)\overline{\phi(s)} \frac{\hr}{R}\right|^2 = \left(\frac{\hr}{R}\right)^2,\]
where the penultimate equality uses $\hat{g}(\phi) = \phi(s) \hat{f}(\phi)$ from \cref{lem:ghat}.
\end{proof}

\subsection{Second generalization: more parameters}\label{subsec: second gen}

Recall that \cref{def: hr R bounded intro} requires an upper bound $R$ on $f$ and a lower bound $\hr$ on $\smash\hf$. This is somewhat restrictive since, for example, no zeros are allowed in the Fourier spectrum $\smash\hf$. Moreover, even if $f$ is very large or $\smash\hf$ very small at a single point, the ratio $\hr/R$ and hence the success probability \eqref{eq:basic p(s)} is very small.
This motivates a generalization to $(r, R, \hr, \hR, A, \hA)$-bounded functions as in \cref{def: r R hr hR bounded intro}.
\setcounter{footnote}{0}

In this section, we give a quantum algorithm for the hidden shift problem of such functions which is very similar to \cref{alg: approx1}. The main difference is that we include the indicator registers of $O_g$ and $\smash{O_{\hf}}$ (i.e., the second register in \cref{eq: oracle access g,eq: oracle access hf}) and post-select\footnote{Post-selection means that we terminate the algorithm if the measurement does not produce the desired outcome. Both intermediate measurements in \cref{alg: approx2} post-select the qubit to state $\ket{1}$.} them to $\ket{1}$. This has the effect of restricting the uniform superpositions over $G$ and $\hG$ to subsets $A + s$ and $\hA$ so that $|g(x)| \in [r, R]$ and $|\hf(\phi)| \in [\hr,\hR]$, respectively (see \cref{def: r R hr hR bounded intro}). The algorithm works on the space $\C^G \otimes \C^2 \otimes \C^{2^n} \otimes \C^2 \otimes \C^2$, where the second register is new compared to the previous section. Other than that here we use the same sequence of operations as \cref{alg: approx1}.

\begin{Algo}~\label{alg: approx2}
\begin{enumerate}
    \item Prepare the uniform superposition $\ket{G,0,0,0,0}:=\frac{1}{|G|^{1/2}} \sum_{x\in G} \ket{x,0,0,0,0}$.
    \item Perform the following quantum circuit:
    
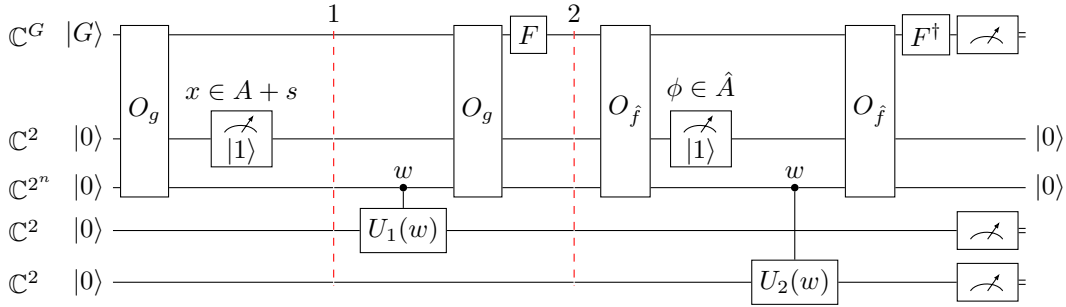
\begin{figure}[H]
    \centering
    \begin{tikzpicture}
        \def\dx{1.5cm} 
    	\begin{yquant}
            [name = regG] qubit {$\ket{G}$} G;
            [name = regO] qubit {$\ket{0}$} O;
            [name = regw] qubit {$\ket{0}$} w;
            [name = rega] qubit {$\ket{0}$} a;
            [name = regb] qubit {$\ket{0}$} b;
            \path (regG.east)+(-\dx,0) node[anchor=west] {$\C^G$};
            \path (regO.east)+(-\dx,0) node[anchor=west] {$\C^2$};
            \path (regw.east)+(-\dx,0) node[anchor=west] {$\C^{2^n}$};
            \path (rega.east)+(-\dx,0) node[anchor=west] {$\C^2$};
            \path (regb.east)+(-\dx,0) node[anchor=west] {$\C^2$};

    		box {$O_g$} (G,O,w);
    		[value = $\ket{1}$, "$x \in A+s$" above] measure O;
            settype {qubit} O;

            [red, label = $1$] barrier (-);
            [name = U1] box {$U_1(w)$} a | w;
            \path (U1-p) node[above] {$w$};
    		box {$O_g$} (G,O,w);
    		box {$F$} G;

            [red, label = $2$] barrier (-);
    		box {$O_{\hf}$} (G,O,w);
            [value = $\ket{1}$, "$\phi \in \hA$" above] measure O;
    		settype {qubit} O;

    		align O,w;
            [name = U2] box {$U_2(w)$} b | w;
            \path (U2-p) node[above] {$w$};
    		box {$O_{\hf}$} (G,O,w);

            box {$F^{\dagger}$} G;

    		align G,a,b;
    		measure G;
    		measure a;
            measure b;

            output {$\ket{0}$} O;
            output {$\ket{0}$} w;
        \end{yquant} 
    \end{tikzpicture}
    \end{figure}
    Output \textnormal{FAIL} if either of the two post-selections on the second register fail to produce~$\ket{1}$.
    \item Measure the first, fourth and fifth register in the standard basis. If the measurements on the last two registers result in~$0$, output the value of the first register. Otherwise output \textnormal{FAIL}.
\end{enumerate}
\end{Algo}

\begin{Thm}\label{thm: second gen}
    For an $(r,R, \hr, \hR, \alpha, \halpha)$-bounded function $f: G\to \C$, the probability of finding the hidden shift with \cref{alg: approx2} is 	
	\begin{align}\label{eq: full probability}
		p(s) =&\left(\frac{\hr}{R}\right)^2\left|\halpha - \frac{1}{|G|^{3/2}}\sum_{\phi \in \smash\hA}\sum_{x \notin A+s} \phi(x)\overline{\phi(s)} \frac{g(x)}{\hf(\phi)}\right|^2.
	\end{align}
    The algorithm uses two calls to each of the oracles $O_g$ and $O_{\hf}$.
\end{Thm}

\begin{proof}
Time point~$1$ is reached with probability $\alpha$ and leads to state
$$\frac{1}{|A|^{1/2}}\sum_{x \in A+s} \ket{x,1,g(x),0,0}.$$
Next, at time point~$2$ the state changes to

\begin{equation}\label{eq: ket Phi}
    \ket{\Phi(s)} = \frac{1}{|A|^{1/2}|G|^{1/2}}
    \sum_{\phi \in \smash\hG} \sum_{x \in A+s}
    \phi(x) \ket{\phi,0,0}
    \left(
        \frac{g(x)}{R} \ket{0}
        + \sqrt{1-\left|\frac{g(x)}{R}\right|^2} \ket{1}
    \right)
    \ket{0}.
\end{equation}

Similar to the proof of \cref{thm: first gen}, let us now analyze the circuit backwards. At the end of the circuit the second and third register are always restored to $\ket{0}$ because both oracles are called twice, so the final state should be $\ket{s,0,0,0,0}$ for the algorithm to successfully find the hidden shift $s$. Starting with this state and running the circuit backwards, time point~$2$ is reached with probability $\halpha$ leading to state
\begin{equation}\label{eq: ket Psi}
    \ket{\Psi(s)}  = \frac{1}{|\hA|^{1/2}}\sum_{\phi \in \smash\hA} \phi(s) \ket{\phi,0,0,0}\left(\frac{\hr}{\overline{\hf(\phi)}}\ket{0} + \sqrt{1-\left|\frac{\hr}{\overline{\hf(\phi)}}\right|^2}\ket{1}\right).
\end{equation}

The total probability of finding the shift $s$ can now be computed as
\begin{align}
    p(s)
    &= \alpha\halpha\left|\bk{\Psi(s)}{\Phi(s)}\right|^2 \nonumber \\
    &= \alpha\halpha \left| \frac{1}{|\hA|^{1/2}|A|^{1/2}|G|^{1/2}} \sum_{\phi \in \smash\hA}\sum_{x \in A+s} \phi(x)\overline{\phi(s)}\frac{g(x)\hr}{\hf(\phi)R}\right|^2 \nonumber \\
    &= \left(\frac{\hr}{R}\right)^2 \left|\frac{1}{|G|^{3/2}} \sum_{\phi \in \smash\hA} \sum_{x \in A+s} \phi(x) \overline{\phi(s)} \frac{g(x)}{\hf(\phi)}\right|^2.
    \label{eq:p(s) double sum}
\end{align}
The second summation in \cref{eq:p(s) double sum} can be split into two parts:
\begin{align*}
    \frac{1}{|G|^{3/2}} \sum_{\phi \in \smash\hA} \sum_{x \in A+s} \phi(x) \overline{\phi(s)} \frac{g(x)}{\hf(\phi)}
    &= \frac{1}{|G|^{3/2}}\sum_{\phi \in \smash\hA} \left(\sum_{x\in G} - \sum_{x\notin A+s}\right)\left(\phi(x)\overline{\phi(s)} \frac{g(x)}{\hf(\phi)}\right)\\
    &= \frac{1}{|G|}\sum_{\phi \in \smash\hA}\left(\overline{\phi(s)}\frac{\hg(\phi)}{\hf(\phi)} - \frac{1}{|G|^{1/2}}\sum_{x\notin A+s} \phi(x)\overline{\phi(s)} \frac{g(x)}{\hf(\phi)} \right)\\
    &= \halpha - \frac{1}{|G|^{3/2}}\sum_{\phi \in \smash\hA}\sum_{x \notin A+s} \phi(x)\overline{\phi(s)} \frac{g(x)}{\hf(\phi)}.
\end{align*}
Plugging this back into \cref{eq:p(s) double sum} we get the desired expression.
\end{proof}

The expression in \cref{thm: second gen} reduces to that of \cref{thm: first gen} if we set $r=0$, $\hR=\infty$ and require $A=G$, $\hA=\hG$. Then $\halpha=1$ and the double summation in \cref{eq: full probability} vanishes, giving us the probability in \cref{eq:basic p(s)} from the first generalization.

\subsection{Example: characters}

As an application of \cref{thm: second gen}, let us calculate the success probability of \cref{alg: approx2} for two examples: \emph{primitive Dirichlet characters} and \emph{finite field characters}.

\paragraph{Primitive Dirichlet characters.}
Let $G = \Z/n\Z$ be the additive group of integers modulo $n$ and consider a multiplicative character $f: \Z/n\Z \to \C$. That is, $f$ is a group homomorphism on the multiplicative group of units $(\Z/n\Z)^*$, extended to $\Z/n\Z$ by defining it to be zero elsewhere. 
We call $f$ \emph{imprimitive} when there is a divisor $n_1\mid n$ and a multiplicative character $f_1: (\Z/n_1\Z)^*\to \C$ such that 
\[f(x) = f_1(x\bmod{n_1}) \text{ for all } x \in (\Z/n\Z)^*,\]
otherwise we call $f$ \emph{primitive}.

To solve the hidden shift problem for a primitive character $f$ we need to find the necessary parameters in \cref{def: r R hr hR bounded intro}. Setting $A = (\Z/n\Z)^*$ we have
\[|f(x)| = \begin{cases}
    1 &\text{ if $x \in A$,}\\
    0 &\text{ otherwise.}
\end{cases}\]
It holds that $|A| = \varphi(n)$, the Euler's totient function, so we take $\alpha = \varphi(n)/n$ and $r=R=1$. For the Fourier transform we make the identification
$$G \cong \hG,\: y \mapsto (\phi_y: x \mapsto e^{2\pi ixy/n}).$$
For any $y\in (\Z/n\Z)^*$, using $\phi_y(x) = \phi_1(xy)$, $Ay = A$, and the multiplicativity of $f$
we find that
\begin{align}
    \hat{f}(y)
    & = \frac{1}{|G|^{1/2}} \sum_{x \in A} f(x) \phi_y(x)\nonumber\\
    & = \frac{1}{|G|^{1/2}} \sum_{x' \in Ay}f(x'y^{-1})\phi_1(x')\nonumber\\
    & = \frac{1}{|G|^{1/2}} \overline{f(y)} \sum_{x' \in A} f(x') \phi_1(x')\nonumber\\
    & = \overline{f(y)} \hat{f}(1).\label{eq:f hat and f}
\end{align}
Thus $|\hat{f}(y)|=|\hat{f}(1)|=1$ whenever $\gcd(y,n)=1$. By Parseval's identity \eqref{eq:Parseval} it follows that $\hf(y)=0$ otherwise. We can set $\hA= (\Z/n\Z)^*$ with size $|\hA|=\varphi(n)$, making $\halpha = \varphi(n)/n$ and $\hr = \hR=1$. By \cref{thm: second gen}, the success probability of finding $s$ is then $$p(s)= \left(\frac{\hr}{R}\right)^2\left|\halpha - \frac{1}{|G|^{3/2}}\sum_{y \in \smash\hA}\sum_{x \notin A+s} \phi_y(x)\overline{\phi_y(s)} \frac{g(x)}{\hf(y)}\right|^2=\halpha^2 = \left(\frac{\varphi(n)}{n}\right)^2.$$
Note that the double sum vanishes because $g(x)=f(x-s)=0$ whenever $x\notin A+s$.

\paragraph{Finite field characters.}
Let $q=p^k$ be a prime power and $\F_q$ be the finite field with $q$ elements. Consider a multiplicative character $f: \F_q^*\to \C$ that is extended to $\F_q$ with $f(0)=0$. To apply \cref{alg: approx2}, we let $A = \F_q^*$ such that $f$ is zero outside $A$ and we can set $r=R=1$. 

The additive characters of $\F_q$ are given by $\phi_y: x\mapsto e^{2\pi i \Tr(xy)/p}$, where $\Tr: \F_q \to \F_p$ is the trace map given by $\Tr(x) = \sum_{j=0}^{k-1}x^{p^j}$. By the same argument as for Dirichlet characters we know that $\hf(\phi_y) = \overline{f(y)}\hf(1)$, meaning we can also take $\hA=\F_q^*$ with $\hr=\hR=1$ and $\hf$ being zero outside $\hA$. The probability of finding the hidden shift $s$ by \cref{alg: approx2} is thus given by $p(s) = \halpha^2 = (1-\frac{1}{q})^2$. This is the same probability as obtained in \cite{van_Dam_2006}.

\subsection{Symmetry of the algorithm}\label{subsec: symmetry}

The success probability in \cref{thm: second gen} is not symmetric in the parameters $r, R, \hr, \hR, \alpha, \halpha$. Here we explain why and give a mirrored version of \cref{alg: approx2}.

Whenever $g$ and $f$ are not completely bent functions, we introduced parameters $R$ and $\hr$ to separate the superposition into the good and bad part, see \cref{eq: phase explanation}. To get the amplitude as close to $\hg(\phi)/\hf(\phi) = \hat{\delta}_s(\phi) = \phi(s)$ as possible (and thereby learn $s$) we used the approximation 
\[\hat{\delta}_s(\phi) \approx \frac{\hg(\phi)}{R} \frac{\hr}{\hf(\phi)}.\]
We hereby bounded $g$ from above and $\hf$ from below. One can instead use the inverse of this approximation given by
\[\hat{\delta}_s(\phi)^{-1} \approx \frac{r}{\hg(\phi)}\frac{\hf(\phi)}{\hR},\]
where we need to bound $g$ and $\hf$ from the other side. To flip the roles of $\hf$ and $g$ and therefore to approximate $\phi(s)^{-1}= \overline{\phi(s)}$, we can run \cref{alg: approx2} with the following changes.

\begin{itemize}
    \item We replace $U_1$ from \cref{eq: U1 second gen} with $U_1 = \sum_w \proj{w} \x U_1(w)$,
    where
	\begin{equation*}
	U_1(w)\ket{0} = \begin{cases}\frac{r}{w}\ket{0} + \sqrt{1-|r/w|^2}\ket{1} &\text{ if } w\neq 0,\\
    \ket{0}&\text{ if } w=0.
    \end{cases}
	\end{equation*} 
    \item We replace $U_2$ from \cref{eq: U2 second gen} with $U_2 = \sum_w \proj{w} \x U_2(w)$
    where
	\begin{equation*}
	U_2(w)^{\dagger}\ket{0} = \frac{\bar{w}}{\hR}\ket{0} + \sqrt{1-|w/\hR|^2}\ket{1}.
	\end{equation*}
    \item We apply on the first register between the two calls of $O_g$ and also before the final measurement an extra operator $U_{-} \in \U{\C^G}$ acting via
    \[U_{-}\ket{x} = \ket{-x} \text{ for all } x\in G.\]
    The rest of the algorithm stays the same.
\end{itemize}
We can then redo the proof of \cref{thm: second gen} up to \cref{eq:p(s) double sum} by making the changes
\[\frac{g(x)}{R} \mapsto \frac{r}{g(-x)}, \quad \frac{\hr}{\hf(\phi)} \mapsto \frac{\hf(\phi)}{\hR}.\]
Afterwards we can use $\hA\subset \hG$ to split the sum and find that the the modified algorithm has success probability
\begin{equation*}
    p(s)
    =\left(\frac{r}{\hR}\right)^2\left|\alpha - \frac{1}{|G|^{3/2}}\sum_{\phi \notin \smash\hA}\sum_{-x \in A+s} \phi(x)\overline{\phi(-s)}\frac{\hf(\phi)}{g(-x)}\right|^2.
\end{equation*}

\begin{Rem}
    This symmetry is part of the reason we include all parameters $(r, R, \hr, \hR, \alpha, \halpha)$, even though not all of them appear in the success probability. Another (slightly hidden) dependence of parameters in \cref{eq: full probability} lies in the summations over $A+s$ and $\hA$, which are determined by $\alpha$ and $\halpha$ respectively. Lastly, all the parameters can be used to give an upper bound on the double summation in \cref{eq: full probability}.
\end{Rem}

\section{Multidimensional bent functions}\label{sec: multi}

An even greater generalization of bent functions was proposed by Poinsot \cite{poinsot2005multidimensional} which includes functions from a finite abelian group to any hermitian space. We restrict ourselves to functions with codomain $\C^d$ for some dimension $d\geq 1$.

\begin{Def}[Multidimensional Fourier transform] \label{Def:multidim Fourier}
    Let $f: G\to \C^d$ be a multidimensional complex function. Writing the function coordinate-wise as $f(x) = (f_0(x), \dots, f_{d-1}(x))$, the Fourier transform of $f$ is given by 
    \[
        \hf(\phi) = (\hf_0(\phi), \dots, \hf_{d-1}(\phi)) \text{ for all } \phi \in \hG,
    \]
    where each $\hf_i$ is the one-dimensional Fourier transform from \cref{def: Fourier tranform}.
\end{Def}

\begin{Def}[Multidimensional bent function]\label{def:multidim bent}
    A multidimensional complex function $f: G \to \C^d$ is called \textit{bent} if $\norm{f(x)} = 1$ for all $x\in G$ and $\norm{\hf(\phi)} = 1$ for all $\phi \in \hG$.
\end{Def}

\begin{Ex}\label{Ex: 2-dim bent function}
While a multidimensional bent function has unit norm everywhere, this may fail for its one-dimensional parts.
For example, writing $\omega = e^{2\pi i/3}$, this happens for the function $f = (f_0, f_1)$ where $f_i : \Z/3\Z \to \C$ are given by
\[\begin{array}{c|ccc}
    &0&1&2\\
    \hline
    f_0(x)&1&\frac{\omega+\omega^2}{2}&1\\
    f_1(x)&0&\frac{\omega-\omega^2}{2}&0
\end{array} \qquad
\begin{array}{c|ccc}
    &0&1&2\\
    \hline
    \hf_0(\phi)&\frac{\sqrt{3}}{2}&-\frac{\sqrt{3}\omega}{2}&-\frac{\sqrt{3}\omega^2}{2}\\
    \hf_1(\phi)&\frac{\omega-\omega^2}{2\sqrt{3}}&\frac{\omega^2-1}{2\sqrt{3}}&\frac{1-\omega}{2\sqrt{3}}
\end{array}\quad.\]
\end{Ex}
See \cref{apx:multidim bent} for more discussion and examples.

\subsection{Exact algorithm}

Consider a multidimensional bent function $f:G \to \C^d$ on a finite abelian group $G$. A function $g: G \to \C^d$ hides a shift $s\in G$ if $g(x) = f(x-s)$ for all $x\in G$.
We present an exact quantum algorithm, similar to \cref{algo1}, to solve the hidden shift problem, assuming that we have access to the quantum oracles $O_g$ and $O_{\hf}$ given by \cref{eq: oracle access g,eq: oracle access hf}.
One ingredient will be unitary operators $\tilde{O}_g \in \U{\C^G\otimes \C^d}$ and $\tilde{O}_{\hf}\in \U{\C^{\hG}\otimes \C^d}$ satisfying
\[
\tilde{O}_g \ket{x} \ket{0} = \ket{x} \sum_{i=0}^{d-1} g_i(x)\ket{i}, \qquad
\tilde{O}_{\hf} \ket{\phi} \ket{0} = \ket{\phi} \sum_{i=0}^{d-1} \hf_i(\phi)\ket{i}.
\]
One can view $\tilde{O}_g$ and $\tilde{O}_{\hf}\ct$ as multidimensional versions of the phase oracles defined in \cref{eq: phase oracles} and used in \cref{algo1}.

\begin{Lem}
    Let $f: G\to\C^d$ be a multidimensional complex function such that $\norm{f(x)} = 1$ for all $x \in G$. Given access to oracle $O_f \in \U{\C^G \otimes \C^{2^n}}$ implementing the transformation $\ket{x}\ket{b}\mapsto \ket{x}\ket{b \oplus f(x)}$ for all $x \in G$ and $b \in \set{0,1}^n$, one can implement an operator $\tilde{O}_f \in \U{\C^G \otimes \C^d}$ satisfying $\tilde{O}_f \ket{x} \ket{0} = \ket{x} \sum_{i=0}^{d-1} f_i(x)\ket{i}$ using two calls to $O_f$.
\end{Lem}
\begin{proof}
We start from the state $\ket{x}\ket{0}\ket{0} \in \C^G \x \C^{2^n} \x \C^d$, where the $n$-qubit register in the middle stores $f(x)$ as an $n$-bit string. Recall that we assumed in \cref{sec:problem} to have a perfect encoding of all values $f(x) \in \C^d$ using $n$ qubits. We call $O_f$ to obtain the state $\ket{x}\ket{f(x)}\ket{0}$. In a new auxiliary register, we then compute a classical description of a unitary operator $S\in \U{\C^d}$, expressed as a sequence of Givens rotations, satisfying $S\ket{0}=\sum_{i=0}^{d-1}f_i(x)\ket{i}$.
Through a sequence of controlled operations, $S$ is applied to the quantum state of the third register, after which we uncompute the data describing $S$ in the auxiliary register.
This results in the state $\ket{x}\ket{f(x)}\sum_{i=0}^{d-1}f_i(x)\ket{i}$.
Finally, we call $O_f$ again to uncompute the second register.
\end{proof}

\begin{Algo}~\label{alg: exact multidim}
    \begin{enumerate}
        \item Prepare the uniform superposition $\ket{G}\ket{0} =\frac{1}{|G|^{1/2}}\sum_{x\in G} \ket{x}\ket{0}$.
        \item Apply the operator $(F^{\dagger}\otimes I)\tilde{O}_{\hf}\ct(F\otimes I)\tilde{O}_g$ to transform the state as 
        \begin{align*}
        \frac{1}{|G|^{1/2}}\sum_{x\in G}\ket{x}\ket{0} &\xmapsto{\tilde{O}_g} \frac{1}{|G|^{1/2}}\sum_{x\in G}\ket{x}\otimes \left(\sum_{i=0}^{d-1} g_i(x)\ket{i}\right)\\
        &\xmapsto[]{F\otimes I} \frac{1}{|G|^{1/2}}\sum_{\phi\in \hG}\ket{\phi} \otimes \phi(s)\left(\sum_{i=0}^{d-1} \hf_i(x)\ket{i}\right)\\
        &\xmapsto{\tilde{O}_{\hf}\ct} \frac{1}{|G|^{1/2}}\sum_{\phi\in \hG}\ket{\phi} \otimes \phi(s) \ket{0}\xmapsto[]{F^{\dagger}\otimes I} \ket{s}\ket{0}.
        \end{align*}
        \item Measure the first register in the standard basis.
    \end{enumerate}
\end{Algo}

\begin{Thm}
    For a multidimensional bent function $f: G\to \C^d$, \cref{alg: exact multidim} determines the hidden shift $s$ with certainty. 
\end{Thm}

\begin{Rem}
    The above algorithm also covers the one-dimensional case from \cref{algo1}. By taking $d=1$ the values $g(x)$ and $\hf(\phi)$ lie on the unit circle for all $x\in G, \phi\in \hG$.
    Then the second register can be discarded, so $\tilde{O}_g$ and $\tilde{O}_{\hf}\ct$ reduce to the phase oracles from~\cref{eq: phase oracles}. The algorithm then works the same way and returns the hidden shift with certainty.
\end{Rem}

\subsection{Approximate algorithm}\label{subsec: multi gen}

We have generalized the one-dimensional \cref{algo1} for bent functions in two different ways.
\begin{enumerate}
    \item \Cref{alg: approx2} allows for function values outside the unit circle. This was done by re-normalizing the function and preparing an ancillary qubit in a superposition of a `good' part and a `bad' part:
    \[\underbrace{\vphantom{\sqrt{1-\left|\frac{g(x)}{R}\right|^2}\ket{1}}\frac{g(x)}{R} \ket{0}}_{\text{good}} + \underbrace{\sqrt{1-\left|\frac{g(x)}{R}\right|^2}\ket{1}}_{\text{bad}}.\]
    \item \Cref{alg: exact multidim} allows for multidimensional bent functions. Where \cref{algo1} used $g(x)$ and $\hf(\phi)$ as global phases, i.e.\ one-dimensional rotations, we now see the vector valued functions as rotations in $\C^d$.
\end{enumerate}
We will now combine these two generalizations and describe an approximate algorithm for bent-like functions in higher dimensions. The approximate algorithm is therefore very similar to \cref{alg: approx2}, but with an added rotation to reduce to the one-dimensional case. We assume the functions to satisfy \cref{def: multidimensional all bounded}. 

We work with the state space $\C^{G}\otimes \C^{2} \otimes \C^{2^n} \otimes \C^d \otimes \C^2$, where the third register acts as extra workspace to store the vector of function values and will be returned to its original state by the algorithm.
We assume that we have access to the following quantum operators.
\begin{itemize}
    \item Oracle access to the functions $g$ and $\hf$ via \cref{eq: oracle access g,eq: oracle access hf}.
    \item For ease of notation we write $\ket{w_0, \dots, w_{d-1}} = \ket{w}\in \C^{2^n}$. Note the different dimensions of $w\in \C^{d}$ and $\ket{w}\in \C^{2^n}$, as we assumed in \cref{sec:problem} there is a perfect encoding of the $d$-dimensional values $g(x)$ in $n$ qubits. We assume that we have three unitary operators
            \[U_1, V_{\rot} \in \U{\C^{2^n} \otimes \C^d}, U_2 \in \U{\C^{2^n} \otimes \C^{2}}\]
        given by
        \begin{align}
            U_1 &= \sum_w \proj{w} \x U_1(w), \text{ where }
            U_1(w)\ket{0} = \sum_{i=0}^{d-1} \frac{w_i}{R}\ket{i} + \sqrt{1-\frac{\|w\|^2}{R^2}}\ket{d}, \label{eq:U1}\\
            V_{\rot} &= \sum_{w} \ket{w}\bra{w}\otimes V_{\rot}(w), \text{ where }
            V_{\rot}(w)^{\dagger}\ket{0} = \begin{cases}
                \frac{1}{\|w\|}\sum_{i=0}^{d-1}w_i\ket{i}&\text{ if } w\neq 0,\\
                \ket{0} &\text{ if } w=0,
            \end{cases}
            \label{eq:Vrot}\\
            U_2 &= \sum_w \proj{w} \x U_2(w), \text{ where } 
            U_2(w)^{\dagger}\ket{0} = \begin{cases}\frac{\hr}{\|w\|}\ket{0}+\sqrt{1-\frac{\hr^2}{\|w\|^2}}\ket{d} &\text{ if } w\neq 0,\\
            \ket{0} &\text{ if } w = 0.
            \end{cases}
            \label{eq:U2}
        \end{align}
    \item The Fourier transform on the register $\C^{G}$.
\end{itemize}
These operators have many similarities to those in \cref{subsec: second gen}.
The main difference is that the operator from \cref{eq: U1 second gen} is now split into $U_1$ and $V_{\rot}$ to account for the multiple dimensions.
First we will use $U_1$, then apply the Fourier transform in each coordinate and then rotate backwards to put all the information in the first coordinate.
This inverse rotation reduces it to a one-dimensional problem.
The algorithm works on the space $\C^G\otimes \C^2\otimes \C^{2^n}\otimes \C^d\otimes \C^2$ as follows.
\begin{Algo}~\label{alg: generalised bent multid}
\begin{enumerate}
	\item Prepare the uniform superposition $\ket{G,0,0,0,0}:=\frac{1}{|G|^{1/2}} \sum_{x\in G}      \ket{x,0,0,0,0}$.
    \item Perform the following quantum circuit.
    
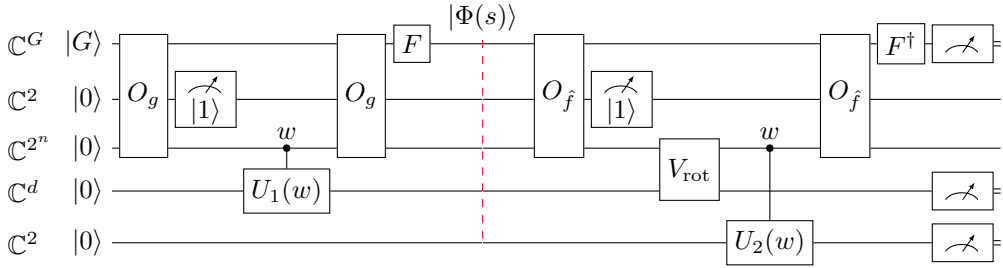
\begin{figure}[H]
    \centering
    \begin{tikzpicture}
        \def\dx{1.5cm} 
    	\begin{yquant}
            [name = regG] qubit {$\ket{G}$} G;
            [name = regO] qubit {$\ket{0}$} O;
            [name = regw] qubit {$\ket{0}$} w;
            [name = rega] qubit {$\ket{0}$} a;
            [name = regb] qubit {$\ket{0}$} b;
            \path (regG.east)+(-\dx,0) node[anchor=west] {$\C^G$};
            \path (regO.east)+(-\dx,0) node[anchor=west] {$\C^2$};
            \path (regw.east)+(-\dx,0) node[anchor=west] {$\C^{2^n}$};
            \path (rega.east)+(-\dx,0) node[anchor=west] {$\C^d$};
            \path (regb.east)+(-\dx,0) node[anchor=west] {$\C^2$};

    		box {$O_g$} (G,O,w);
    		[value = $\ket{1}$] measure O;
    		settype {qubit} O;
            align O, w;

            [name = U1] box {$U_1(w)$} a | w;
            \path (U1-p) node[above] {$w$};
            box {$O_g$} (G,O,w);
            box {$F$} (G);

            [red, label=$\ket{\Phi(s)}$] barrier (G,O,w,a,b);
    		box {$O_{\hf}$} (G,O,w);
    		[value = $\ket{1}$] measure O;
    		settype {qubit} O;
            align O, w;

    		box {$V_{\rot}$} (w,a);
            [name = U2] box {$U_2(w)$} b | w;
            \path (U2-p) node[above] {$w$};
            box {$O_{\hf}$} (G,O,w);
            box {$F^{\dagger}$} (G);

    		align G,a,b;
    		measure G;
    		measure a;
            measure b;
    	\end{yquant} 
    \end{tikzpicture}
    \label{fig: approx multidim}
    \end{figure}
    Output \textnormal{FAIL} if either of the two intermediate measurements that post-select the second register to~$\ket{1}$ fail.
    \item Measure the first, third and fourth register in the standard basis. If the latter two measure $0$, output the value of the first register. Otherwise output \textnormal{FAIL}.
\end{enumerate}
\end{Algo}
Note the similarities with \cref{alg: approx2}.

\begin{Thm}\label{thm: multidim gen}
    For an $(r, R, \hr, \hR, \alpha, \halpha)$-bounded function $f: G \to \C^d$, the probability of finding the hidden shift $s$ by \cref{alg: generalised bent multid} is 
    \[
        p(s) = \left(\frac{\hr}{R}\right)^2\left|\halpha- \frac{1}{|G|^{3/2}}\sum_{\phi \in \hA}\sum_{x \notin A+s} \phi(x)\sum_{i=0}^{d-1}\frac{g_i(x)\overline{\hf_i(\phi)}}{\|\hf(\phi)\|}\right|^2,
    \]
    and it uses $2$ calls to both the oracles $O_g$ and $O_{\hf}$.
\end{Thm}
\begin{proof}
    The proof is very similar to the one of \cref{thm: second gen}. For ease of notation we denote 
    \[\ket{g_0(1), \dots, g_{d-1}} = \ket{g(x)}, \text{ and } \ket{\hf_0(\phi), \dots, \hf_{d-1}(\phi)} = \ket{\hf(\phi)}.\]
    To correctly identify the hidden shift $s$, we would like to observe the state $\ket{s,0,0,0,0}$ at the end. As we did for the proofs of \cref{thm: first gen,thm: second gen}, we define $\ket{\Phi(s)}$ to be the state at the marked time point (assuming the intermediate measurement do not fail). Let $\ket{\Psi(s)}$ be the state at the same point obtained by running the circuit backwards starting with the target state $\ket{s,0,0,0,0}$ (also assuming the intermediate measurement does not fail). As this time point is reached with probability $\alpha$ and $\halpha$ from the left and right respectively, the total probability is given by 
    \[p(s) = \alpha\halpha \left|\bk{\Psi(s)}{\Phi(s)}\right|^2.\]
    Similarly to \cref{eq: ket Phi} it holds that
    \[\ket{\Phi(s)} = \frac{1}{|A|^{1/2}|G|^{1/2}}
    \sum_{\phi \in \smash\hG} \sum_{x \in A+s}
    \phi(x) \ket{\phi,0,0}
    \left(\sum_{i=0}^{d-1}\frac{g_i(x)}{R} \ket{i} + \sqrt{1-\left|\frac{\|g(x)\|}{R}\right|^2} \ket{d}
    \right)
    \ket{0}.\]
    Starting with the target state $\ket{s,0,0,0,0}$ and reading from the right, we obtain a state similar to \cref{eq: ket Psi}, only with an extra $V_{\rot}^{\dagger}$ applied to the fourth register (conditional on $\hf(\phi)\neq0$). Note that in \cref{thm: second gen} this register is untouched by the right side of the circuit. The result is 
    \[\ket{\Psi(s)} = \frac{1}{|\hA|^{1/2}} \sum_{\phi \in \hA} \phi(s)\ket{\phi} \ket{0,0}\left(\sum_{i=0}^{d-1} \frac{\hf_i(\phi)}{\|\hf(\phi)\|}\ket{i}\right)\left(\frac{\hr}{\|\hf(\phi)\|)}\ket{0} + \sqrt{1 - \left|\frac{\hr}{\|\hf\|}\right|^2}\ket{1}\right).\]
    The total probability is thus given by 
    \begin{align*}
        p(s)&= \alpha\halpha \left|\bk{\Psi(s)}{\Phi(s)}\right|^2\\
        &= \alpha\halpha \left| \frac{1}{|\hA|^{1/2}|A|^{1/2}|G|^{1/2}} \sum_{\phi \in \smash\hA}\sum_{x \in A+} \left(\phi(x)\overline{\phi(s)} \sum_{i=0}^{d-1} \frac{g_i(x)\overline{\hf_i(\phi)}}{\|\hf(\phi)\|^2}\right)\right|^2\\
        &= \left(\frac{\hr}{R}\right)^2\left|\frac{1}{|G|^{3/2}}\sum_{\phi \in \hA}\sum_{x\in G}\left( \phi(x)\overline{\phi(s)}\sum_{i=0}^{d-1}\frac{g_i(x)\overline{\hf_i(\phi)}}{\|\hf(\phi)\|^2}\right) -\right.\\
        &\hspace{6cm} \left. \frac{1}{|G|^{3/2}}\sum_{\phi \in \smash\hA}\sum_{x \notin A+s} \left( \phi(x)\overline{\phi(s)}\sum_{i=0}^{d-1}\frac{g_i(x)\overline{\hf_i(\phi)}}{\|\hf(\phi)\|^2}\right)\right|^2 \\
        &= \left(\frac{\hr}{R}\right)^2\left|\frac{1}{|G|}\sum_{\phi \in \hA}\left(\sum_{i=0}^{d-1}\frac{\overline{\phi(s)}\hg_i(x)\overline{\hf_i(\phi)}}{\|\hf(\phi)\|^2}\right) - \frac{1}{|G|^{3/2}}\sum_{\phi \in \smash\hA}\sum_{x \notin A+s} \left( \overline{\phi(s)}\sum_{i=0}^{d-1}\frac{g_i(x)\overline{\hf_i(\phi)}}{\|\hf(\phi)\|^2}\right)\right|^2  \\
        &= \left(\frac{\hr}{R}\right)^2\left|\halpha - \frac{1}{|G|^{3/2}}\sum_{\phi \in \smash\hA}\sum_{x \notin A+s} \left( \overline{\phi(s)}\sum_{i=0}^{d-1}\frac{g_i(x)\overline{\hf_i(\phi)}}{\|\hf(\phi)\|^2}\right)\right|^2.\qedhere
        \end{align*}
\end{proof}

\section{One-register approach}\label{sec: one register}

Recall from \cref{subsec: first gen} that at the end of \cref{alg: approx1} the measurement on both ancilla registers had to produce $0$ for the algorithm to succeed.
In this section, we consider a variation of this algorithm where we reduce the number of ancilla registers from two to one and include additional degrees of freedom within both oracles.
Since both oracles now use the same ancilla register, this yields a richer interference pattern and also slightly reduces quantum memory.
Here we investigate how this modification affects the overall success probability of the algorithm.

Recall that \cref{alg: approx1} from \cref{subsec: first gen} prepared states
\[\frac{g(x)}{R}\ket{0} + \sqrt{1-\left(\frac{g(x)}{R}\right)^2} \ket{1} \quad  \text{and }\quad  \frac{\hr}{\hf(\phi)}\ket{0} +  \sqrt{1-\left(\frac{\hr}{\hf(\phi)}\right)^2} \ket{1}\]
on two different registers, and for both states we want as much amplitude as possible in the first basis state $\ket{0}$. However, we can try to improve the success probability of \cref{alg: approx1} by also making use of what was considered in \cref{eq: phase explanation} as error or the `bad part'. The modified algorithm works by letting the unitary operators $U_1$ and $U_2$ from \cref{eq: U1 second gen,eq: U2 second gen} act on the same ancilla register, see \cref{alg:approx}. Previously these operators did not interact directly, so some of their degrees of freedom were not relevant, whereas now they could potentially be used to our advantage. We study the success probability of \cref{alg:approx} in general and determine necessary and sufficient conditions for it to be equal to~1.

\subsection{Oracles}\label{subsec:oracles}

Assume that $f$ is an $(R,\hr)$-bounded function, and that we have quantum access to the oracles of $g$ and $\hf$ as in \cref{eq: oracle access g,eq: oracle access hf}.
Combining $O_g$ with the unitary $U_1$ given in \cref{eq: U1 second gen}, we obtain a unitary oracle $U_g \in \U{\C^G\otimes \C^2}$ that is given by
\begin{equation}\label{eq:oracle U_g}
    U_g
    = \sum_{x \in G} \proj{x} \x
    \mx{b_0(x) & * \\ b_1(x) & *}, \qquad b_0(x) = \frac{g(x)}{R}, \:b_1(x) = e^{i\theta(x)} \sqrt{1-\abs*{\frac{g(x)}{R}}^2},
\end{equation}
where we have introduced a phase function $\theta: G \to [0.2\pi)$ to take advantage of the additional degree of freedom.
The remaining entries indicated by $*$ in \cref{eq:oracle U_g} can be chosen arbitrarily, so long as each $2 \times 2$ block is unitary.
The actual choice (which involves an additional phase degree of freedom) will not matter to us.
A similar construction can be done using $O_{\hf}$ to obtain a unitary $U_{1/\hf} \in \U{\C^{\hG} \otimes \C^2}$ given by
\begin{equation}\label{eq: oracle U_f}
    U_{1/\hat{f}}
    = \sum_{\phi \in \Ghat} \proj{\phi} \x
    \mx{a_0(\phi) & a_1(\phi) \\ * & *}, \qquad a_0(\phi) = \frac{\hr}{\hat{f}(\phi)},\: a_1(\phi) = e^{-i\chi(\phi)} \sqrt{1-\abs*{\frac{\hr}{\hat{f}(\phi)}}^2},
\end{equation}
where we also introduced a phase function $\phi: \hG\to [0, 2\pi)$. For convenience, we let
\begin{align*}
    \bra{a(\phi)} &= \mx{a_0(\phi) & a_1(\phi)}, &
    \ket{b(x)} &= \mx{b_0(x) \\ b_1(x)}.
\end{align*}
Then the action of the oracles from \cref{eq:oracle U_g,eq: oracle U_f} is captured by
\begin{align}
    \bra{\phi} \bra{0} \cdot U_{1/\hat{f}}
    &= \bra{\phi} \bra{a(\phi)}
     = \bra{\phi} \of*{
            \frac{\hr}{\hat{f}(\phi)} \bra{0}
            + e^{-i\chi(\phi)} \sqrt{1-\abs*{\frac{\hr}{\hat{f}(\phi)}}^2} \bra{1}
        }
    & \text{for all } \phi &\in \Ghat,
    \label{eq:backwards action} \\
    U_g \cdot \ket{x} \ket{0}
    &= \ket{x} \ket{b(x)}
     = \ket{x} \of*{
            \frac{g(x)}{R} \ket{0}
            + e^{i\theta(x)} \sqrt{1-\abs*{\frac{g(x)}{R}}^2} \ket{1}
        }
    & \text{for all } x &\in G.
    \label{eq:forward action}
\end{align}

\subsection{Approximate algorithm}

The algorithm for solving the hidden shift problem for an $(R,\hr)$-bounded function $f$ uses the oracles $U_{1/\hat{f}}$ and $U_g$ defined above. The algorithm works on the space $\CG \x \C^2$.

\begin{Algo}~\label{alg:approx}
	\begin{enumerate}
		\item Prepare the state $\ket{G} \ket{0} = \frac{1}{|G|^{1/2}}\sum_{x\in G} \ket{x} \ket{0}$.
		\item Apply the operator $(F^\dagger \x I) U_{1/\hat{f}} (F \x I) U_g$:
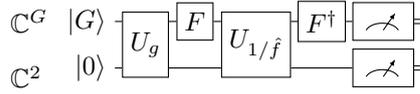
        \begin{center}
        \begin{tikzpicture}
            \def\dx{1.5cm} 
        	\begin{yquant}
                [name = regG] qubit {$\ket{G}$} G;
                [name = rega] qubit {$\ket{0}$} a;
                \path (regG.east)+(-\dx,0) node[anchor=west] {$\C^G$};
                \path (regO.east)+(-\dx,0) node[anchor=west] {$\C^2$};
                
                box {$U_g$} (G,a);
                box {$F$} G;
                box {$U_{1/\hf}$} (G,a);
                box {$F^{\dagger}$} G;
    
                align G, a;
                measure G;
                measure a;
            \end{yquant} 
        \end{tikzpicture}
        \end{center}
        \item Measure both registers in the standard basis.
        If the second register contains $0$, output the value of the first register.  Otherwise, output \textnormal{FAIL}.\label{it:measure}
	\end{enumerate}
\end{Algo}

\begin{Thm}\label{Thm:approximate_algorithm}
    For an $(R, \hr)$-bounded function $f: G\to \C$, the probability of finding the hidden shift $s$ by \cref{alg:approx} is
    \begin{equation}
        p(s)
        = \abs*{
            \frac{\hr}{R}
          + \frac{1}{|G|^{3/2}}
            \sum_{x \in G}
            \sum_{\phi \in \Ghat}
            \phi(x)
            e^{i\theta(x+s)}
            \sqrt{1-\abs*{\frac{f(x)}{R}}^2}
            e^{-i\chi(\phi)}
            \sqrt{1-\abs*{\frac{\hr}{\hat{f}(\phi)}}^2}
        }^2.
        \label{proba}
    \end{equation}
\end{Thm}

\begin{proof}
We follow the proof of \cref{thm: first gen}. The success probability of \cref{alg:approx} is given by
\[
    p(s) =
    \abs[\Big]{
        \underbrace{
            \of[\big]{\bra{s} \x \bra{0}}
            (F\ct \x I) U_{1/\hat{f}}
        }_{\bra{\Psi(s)}}
        \underbrace{\vphantom{U_{1/\hat{f}}}
            (F \x I) U_g
            \of[\big]{\ket{G} \x \ket{0}}
        }_{\ket{\Phi(s)}}
    }^2,
\]
which we will compute by separately evaluating $\bra{\Psi(s)}$ and $\ket{\Phi(s)}$.
By \cref{eq:backwards action,eq:forward action}, we have
\begin{align*}
    \bra{\Psi(s)}
   &= \of[\big]{\bra{s} F\ct \x \bra{0}}
      U_{1/\hat{f}}
    = \frac{1}{|G|^{1/2}}
      \sum_{\phi \in \Ghat}
      \overline{\phi(s)}
      \bra{\phi} \x \bra{a(\phi)}, \\
    \ket{\Phi(s)}
   &= (F \x I) U_g
      \of*{\frac{1}{|G|^{1/2}}
      \sum_{x \in G} \ket{x} \x \ket{0}}
    = \frac{1}{|G|^{1/2}}
      \sum_{x \in G} F \ket{x} \x \ket{b(x)}.
\end{align*}
Recall from \cref{eq:F} that $\bra{\phi} F \ket{x} = \phi(x) / |G|^{1/2}$. Combined with the expressions for $\bra{a(\phi)}$ and $\ket{b(x)}$ from \cref{eq: oracle U_f,eq:oracle U_g} the inner product between these two states is
\begin{align*}
    \bk{\Psi(s)}{\Phi(s)}
    &= \frac{1}{|G|}
      \sum_{x \in G}
      \sum_{\phi \in \Ghat}
      \overline{\phi(s)}
      \bra{\phi} F \ket{x}
      \bk{a(\phi)}{b(x)}=\frac{1}{|G|^{3/2}}
      \sum_{x \in G}
      \sum_{\phi \in \Ghat}
      \overline{\phi(s)}
      \phi(x)
      \bk{a(\phi)}{b(x)}\\
      &=\frac{1}{|G|^{3/2}}
      \sum_{x \in G}
      \sum_{\phi \in \Ghat}
      \overline{\phi(s)}
      \phi(x)\left(\frac{\hr}{R} \frac{g(x)}{\hat{f}(\phi)}
    + e^{i\theta(x)} \sqrt{1-\abs*{\frac{g(x)}{R}}^2}
      e^{-i\chi(\phi)} \sqrt{1-\abs*{\frac{\hr}{\hat{f}(\phi)}}^2}\right)\\
      &=\frac{\hr}{R}
      + \frac{1}{|G|^{3/2}}
        \sum_{x \in G}
        \sum_{\phi \in \Ghat}
        \overline{\phi(s)}
        \phi(x)
        e^{i\theta(x)}
        \sqrt{1-\abs*{\frac{g(x)}{R}}^2}
        e^{-i\chi(\phi)}
        \sqrt{1-\abs*{\frac{\hr}{\hat{f}(\phi)}}^2}.
\end{align*}
The last equality follows from the proof of \cref{thm: first gen}. Note that the first term in this expression is the total probability obtained from \cref{alg: approx1}. The desired formula follows by noting that
$\overline{\phi(s)} \phi(x) = \phi(x-s)$ and substituting $x \mapsto x + s$.
\end{proof}

Next, we investigate conditions under which \cref{alg:approx} gives the shift $s$ with certainty.  To this end, we introduce the function
\begin{align}
\label{eq:hs}
h_s:G&\rightarrow\C\\
x&\mapsto e^{i\theta(x+s)}\sqrt{1-\left|\frac{f(x)}{R}\right|^2}\nonumber
\end{align}
and its Fourier transform $\hat h_s : \Ghat\to\C$.
We can express \cref{proba} in terms of $\hat{h}_s$ as
\begin{equation}\label{eq:pZs}
    p(s) = \left|\frac{\hr}{R} + Z_s\right|^2,
\end{equation}
where
\begin{equation}\label{Zs}
    Z_s = \frac{1}{|G|}\sum_{\phi\in \hat{G}} e^{-i\chi(\phi)}\sqrt{1-\left|\frac{\hr}{\hat{f}(\phi)}\right|^2}\hat{h}_s(\phi).
\end{equation}
For later use, note that by definition of $h_s$ and a double application of Parseval's identity \eqref{eq:Parseval} we have  
\begin{equation}\label{hs-norm}
    \sum_{\phi \in \hat{G}} |\hat{h}_s(\phi)|^2 =\sum_{x \in G} \left(1-\left|\frac{f(x)}{R}\right|^2\right)
    = |G| - \frac{1}{R^2} \sum_{\phi \in \hat{G}} \left|\hat{f}(\phi)\right|^2.
\end{equation}

\begin{Thm}
\label{thm-prob1}
Suppose the hidden shift equals $s$.  Then the probability that \cref{alg:approx} outputs the correct shift (namely, $s$) equals 1 if and only if one of the following conditions holds:
\begin{enumerate}
\item $\hr=R$;
\item $\hr<R$, and it holds that
\begin{equation}\label{fhat-Rr}
|\hat f(\phi)|=\sqrt{\hr R}\quad
\text{for all }\phi\in\Ghat,
\end{equation}
and
\begin{equation}\label{hshat}
\hat h_s(\phi) = e^{i\chi(\phi)}\sqrt{1-\frac{\hr}{R}}
\quad\text{for all }\phi\in\Ghat.
\end{equation}
\end{enumerate}
\end{Thm}

\begin{proof}
By \cref{prop:bounded}(\ref{item1}) we have $\hr\le R$.  In case equality holds, \cref{prop:bounded}(\ref{item1}) gives $|f(x)|=R$ and $|\hat f(\phi)|=\hr$ for all $x\in G$ and $\phi\in\Ghat$.  By \cref{proba}, we immediately deduce $p(s)=1$. Now assume that we are in the case $\hr<R$.
If for all $\phi\in\Ghat$ we have $|\hat f(\phi)|=\sqrt{\hr R}$ and $\hat h_s(\phi)=e^{i\chi(\phi)}\sqrt{1-\frac{\hr}{R}}$, then by \cref{eq:pZs,Zs} we obtain $Z_s=1-\hr/R$ and hence $p(s)=1$.

Conversely, assume $p(s)=1$, or equivalently
\begin{equation}\label{Zs-rR}
Z_s=1-\hr/R.
\end{equation}
The Cauchy--Schwarz inequality together with \cref{hs-norm} implies
\begin{align*}
    |Z_s|^2 &\leq \left(
    \frac{1}{|G|}
    \sum_{\phi \in \hat{G}}
    \left(1-\left|\frac{\hr}{\hat{f}(\phi)}\right|^2\right)
    \right)
    \left(
    \frac{1}{|G|}
    \sum_{\phi \in \hat{G}}
    |\hat{h}_s(\phi)|^2
    \right)\\
&= \left(
    1 - \frac{\hr^2}{|G|}
    \sum_{\phi \in \hat{G}}
    \frac{1}{|\hat{f}(\phi)|^2}
    \right)
    \left(
    1-\frac{1}{|G| \cdot R^2}
    \sum_{\phi \in \hat{G}}
    \left|\hat{f}(\phi)\right|^2
    \right)\\
    &=\left(1-\frac{\hr^2}{H}\right)
   \left(1-\frac{A}{R^2}\right),
\end{align*}
where
\[
A = \frac{\sum_{\phi \in \hat{G}}|\hat{f}(\phi)|^2}{|G|}
\quad\text{and}\quad
H = \frac{|G|}{\sum_{\phi \in \hat{G}} \frac{1}{|\hat{f}(\phi)|^2}}.
\]
By the arithmetic-harmonic mean inequality one has $H\leq A$, with equality if and only if $|\hat{f}(\phi)|^2$ is constant over all $\phi\in \Ghat$, in which case $|\hat{f}(\phi)|=\sqrt{H}$ for all $\phi\in\Ghat$. Furthermore, since $Z_s=1-\hr/R>0$ and $H\ge \hr^2$, we have $1-\hr^2/H>0$.
Therefore it holds that
\begin{align*}
|Z_s|^2&\le
  \left(1-\frac{\hr^2}{H}\right)
  \left(1-\frac{H}{R^2}\right)=1-\left(\frac{\hr^2}{H}+\frac{H}{R^2}\right)+\frac{\hr^2}{R^2}.
\end{align*}
By the arithmetic-geometric mean inequality we have
\[
\frac{\hr^2}{H}+\frac{H}{R^2} \ge
2\frac{\hr}{R},
\]
with equality if and only if $H=\hr R$.  We obtain
\begin{align*}
    |Z_s|^2 &\le 1 - 2\frac{\hr}{R} + \frac{\hr^2}{R^2}=\left(1 - \frac{\hr}{R}\right)^2.
\end{align*}
In view of \cref{Zs-rR}, it follows that $H$ indeed equals $\hr R$, so \cref{fhat-Rr} holds.

It remains to prove \cref{hshat}.
To do this, we first note that by \cref{Zs,Zs-rR,fhat-Rr} we have
\[
\sqrt{1-\frac{\hr}{R}} =
\frac{1}{|G|}\sum_{\phi\in\Ghat}
e^{-i\chi(\phi)}\hat h_s(\phi).
\]
This implies
\[
\sqrt{1-\frac{\hr}{R}}\le
\frac{1}{|G|}\sum_{\phi\in\Ghat} |\hat h_s(\phi)|,
\]
with equality if and only if $e^{-i\chi(\phi)}\hat h_s(\phi)$ is real and non-negative for all $\phi\in\Ghat$.  Next, the arithmetic-quadratic mean inequality together with \cref{hs-norm,fhat-Rr} gives
\[
1-\frac{\hr}{R}\le
\frac{1}{|G|}\sum_{\phi\in\Ghat} |\hat h_s(\phi)|^2
= 1-\frac{\hr}{R},
\]
with equality if and only if in addition all $|\hat h_s(\phi)|$ are equal, in which case the common value equals $\sqrt{1-\hr/R}$.  This proves \cref{hshat}.
\end{proof}

\begin{Rem}
In case 1 of \cref{thm-prob1}, the function $f$ is bent by \cref{prop:bounded}(\ref{item1}), and \cref{alg:approx} reduces to \cref{algo1}.
The conditions in case 2 of \cref{thm-prob1} are quite remarkable.  Specifially, \cref{fhat-Rr} implies that $f$ is not only $(R,\hr)$-bounded but even $(R,\sqrt{\hr R})$-bounded.
However, we need to treat $f$ as an $(R,\hr)$-bounded function in order to attain succes probability 1 in \cref{alg:approx}.
\end{Rem}

The conditions of \cref{thm-prob1} are dependent on the shift~$s$.
Since this shift is actually the value we want to determine, it is natural to ask under what conditions we can attain $p(s)=1$ for \emph{all} values of~$s$.

\begin{Lem}\label{lem:ft-flat}
Let $u:G\to\R_{\ge0}$ be a function, and let $\hat u:\Ghat\to\C$ be its Fourier transform.  Then $|\hat u(\phi)|$ is constant for all $\phi\in\Ghat$ if and only if there exists $x_0\in G$ such that $u(x)=0$ for all $x\ne x_0$.  In this case we have
\begin{equation}\label{eq:uhat}
\hat u(\phi) = \frac{1}{\sqrt{|G|}}u(x_0)\phi(x_0)
\quad\text{for all }\phi\in\Ghat.
\end{equation}
\end{Lem}

\begin{proof}
Suppose all $|\hat u(\phi)|$ are equal.  Since the claim clearly holds if $u$ is the zero function, we assume that $u$ is not the zero function.  We write
\[
U = \sum_{x\in G}u(x)>0.
\]
Let $\phi\in\Ghat$.  By the definition of the Fourier transform, the triangle inequality and the assumption that $u(x)\ge0$ for all $x\in G$ we can write
\begin{equation} \label{lem_p}
    |\hat{u}(\phi)| = \frac{1}{|G|}\left|\sum_{x\in G}u(x)\phi(x)\right| \leq \frac{U}{|G|} = |\hat{u}(\phi_0)|.
\end{equation}
By assumption, the inequality in \cref{lem_p} is in fact an equality. This implies
\begin{equation*}
    \left|\sum_{x\in G}u(x)\phi(x)\right| = U,
\end{equation*}
so we have a convex combination of $|G|$ points on the unit circle giving a point on the unit circle, namely
\begin{equation*}
    \left|\sum_{x\in G}\frac{u(x)}{U}\phi(x)\right|  = 1.
\end{equation*}
This can only hold for all $\phi\in\Ghat$ if there exists $x_0\in G$ such that $u(x) = 0$ for all $x\ne x_0$. Conversely, if there is $x_0\in G$ such that $u(x)=0$ for all $x\neq x_0$, then \cref{eq:uhat} holds, and hence we have $|\hat u(\phi)|=|u(x_0)|/\sqrt{|G|}$ for all $\phi\in\Ghat$.
\end{proof}

\begin{Thm}\label{thm:prob-all1}
We have $p(s)=1$ for all $s\in G$ if and only if one of the following conditions holds:
\begin{enumerate}
\item
$\hr=R$;
\item\label{all1-cond2}
$1-\frac{1}{|G|}\le\frac{\hr}{R}<1$, and there exist $x_0\in G$ and $\alpha\in\R$ such that for all $x\in G$ and $\phi\in\Ghat$ it holds that
\begin{align*}
|\hat f(\phi)| &= \sqrt{\hr R},\\
|f(x)| &= \begin{cases}
R\sqrt{1-|G|\left(1-\frac{\hr}{R}\right)}& \text{if }x=x_0,\\
R& \text{if }x\ne x_0,
\end{cases}\\
\theta(x) &= \alpha,\\
\chi(\phi) &= \alpha + \arg(\phi(x_0)).
\end{align*}

\end{enumerate}
\end{Thm}

\begin{proof}
In the case $\hr=R$, the claim follows directly from \cref{thm-prob1}.  Now suppose that we are in the case $\hr<R$.  First suppose that the conditions in \cref{all1-cond2} hold.  Then we compute
\[
h_s(x) = \begin{cases}
e^{i\alpha}\sqrt{|G|}\sqrt{1-\frac{\hr}{R}}& \text{if }x=x_0,\\
0& \text{if }x\ne x_0.
\end{cases}
\]
This implies that for all $\phi\in\Ghat$ we have
\begin{align*}
\hat h_s(\phi)
&= \phi(x_0)e^{i\alpha}\sqrt{1-\frac{\hr}{R}}
 = e^{i\chi(\phi)}\sqrt{1-\frac{\hr}{R}}.
\end{align*}
By \cref{thm-prob1}, it follows that $p(s)$ equals 1 for all $s\in G$.

Conversely, suppose $\hr<R$ and $p(s)=1$ for all $s\in G$.  By \cref{thm-prob1}, \cref{fhat-Rr,hshat} are satisfied for all $s\in G$, and we need to show that the conditions in~\cref{all1-cond2} hold.
By \cref{hshat}, the function $\hat h_s$ is non-zero and independent of~$s$, and therefore the same holds for $h_s$.  It follows that $\theta$ is constant, say
\[
\theta(x)=\alpha
\quad\text{for all }x\in G.
\]
Thus we have
\[
h_s(x)=e^{i\alpha}u(x)
\quad\text{with}\quad
u(x)=\sqrt{1-\left|\frac{f(x)}{R}\right|^2}.
\]
By \cref{hshat} the functions $\hat h_s$ and therefore also $\hat u$ have constant absolute value, so \cref{lem:ft-flat} implies that there exists $x_0\in G$ such that
\[
|f(x)| = R
\quad\text{for all }x\ne x_0.
\]
It then follows from Parseval's identity \eqref{eq:Parseval} and \cref{fhat-Rr} that $(|G|-1)R^2\ge |G|\hr R$, or equivalently
\[
\frac{\hr}{R}\ge 1-\frac{1}{|G|},
\]
and
\[
|f(x_0)| = R\sqrt{1-|G|\left(1-\frac{\hr}{R}\right)}.
\]
We then compute
\[
h_s(x) = \begin{cases}
e^{i\alpha}\sqrt{|G|}\sqrt{1-\frac{\hr}{R}}& \text{if }x=x_0,\\
0& \text{if }x\ne x_0,
\end{cases}
\]
and subsequently
\begin{align*}
\hat h_s(\phi) &= \frac{1}{\sqrt{|G|}}\phi(x_0)h(x_0)
= \phi(x_0)e^{i\alpha}\sqrt{1-\frac{\hr}{R}}.
\end{align*}
Comparing this to \cref{hshat}, we obtain
\[
\chi(\phi) = \alpha + \arg(\phi(x_0))
\quad\text{for all }\phi\in\Ghat,
\]
which concludes the proof.
\end{proof}

\begin{Ex}
We give explicit examples of families of functions satisfying the conditions of \cref{thm:prob-all1} for the groups $G=\Z/n\Z$ with $n\in\{2,3\}$.  We identify $\Ghat$ with $G$ via the isomorphism $\phi:\Z/n\Z\to\Ghat$ defined by $\phi(a)(x) = \exp(2\pi iax/n)$.

For $n=2$, we fix $\eta\in\C$ with $|\eta|=1$ and $\Re\eta\le0$.  We consider the Fourier transform pair
\[
\begin{tabular}{c|cc}
$x$& 0& 1\\
\hline
\vphantom{\Big|}$f(x)$& $\frac{1+\eta}{\sqrt{2}}$& $\frac{1-\eta}{\sqrt{2}}$\\
\hline
\vphantom{\Big|}$\hat f(x)$& 1& $\eta$
\end{tabular}
\]
This satisfies the conditions of \cref{thm:prob-all1} with $R=1/\hr=|1-\eta|/\sqrt{2}$.

For $n=3$, we fix $\eta\in\C$ with $|\eta|=1$ and $\Re\eta\le1/2$.  We consider the Fourier transform pair
\[
\begin{tabular}{c|ccc}
$x$& 0& 1& 2\\
\hline
\vphantom{\Big|}$f(x)$& $\frac{1+2\eta}{\sqrt{3}}$& $\frac{1-\eta}{\sqrt{3}}$& $\frac{1-\eta}{\sqrt{3}}$\\
\hline
\vphantom{\Big|}$\hat f(x)$& 1& $\eta$& $\eta$
\end{tabular}
\]
This satisfies the conditions of \cref{thm:prob-all1} with $R=1/\hr=|1-\eta|/\sqrt{3}$.
\end{Ex}

\subsection{Lower bounds for the success probability}

We have shown that \cref{alg:approx} can only attain success probability 1 under the conditions of \cref{thm-prob1,thm:prob-all1}.
Since these conditions are rather strict, and those of \cref{thm-prob1} moreover depend on the unknown shift $s$, it is natural to ask for lower bounds on the success probability $p(s)$ given by \cref{proba}, using the freedom that we have in the choice of the functions $\theta$ and~$\chi$.

\begin{Thm}\label{th:lowerbound}
Consider \cref{alg:approx} with $\theta = 0$.
For any value of the hidden shift~$s$, the choice of $\chi$ that maximizes the success probability of measuring $s$ is $\chi(\phi) = \arg(\hat{h}(\phi))$ for all $\phi\in\Ghat$ with $\hat{h}(\phi)\ne0$ (and $\chi(\phi)$ is arbitrary for $\hat{h}(\phi)=0$).
The success probability in this case equals
    \[
        p = \left(\frac{\hr}{R}+|G|^{-1}\sum_{\phi}\sqrt{1-\left|\frac{\hr}{\hat{f}(\phi)}\right|^2}|{\hat h}(\phi)|\right)^2
    \]
    with $h(x) = \sqrt{1-|f(x)/R|^2}$.
\end{Thm}

\begin{proof}
Substituting $\theta=0$, we find that \cref{eq:hs,eq:pZs,Zs} simplify to
\[
p(s) = \left|\frac{\hr}{R} + \frac{1}{|G|}\sum_{\phi\in \hat{G}} e^{-i\chi(\phi)}\sqrt{1-\left|\frac{\hr}{\hat{f}(\phi)}\right|^2}\hat{h}(\phi)\right|^2.
\]
This expression is independent of $s$ and is maximal when each term in the sum is real and positive, which happens for the choice of $\chi$ described in the theorem.
\end{proof}

Note that \cref{th:lowerbound} only gives the optimal choice for $\chi$ when we fix $\theta=0$; we do not claim that these choices are optimal among all choices of $\theta$ and $\chi$.
Furthermore, since this choice for $\chi$ seems difficult to compute efficiently, we will next explore the case where $\theta=0$ and where we let $\chi$ take uniformly random values in $[0,2\pi)$.

\begin{Lem}\label{lem:integral}
For any two complex numbers $a$ and $b$, we have
\begin{equation*}
  \frac{1}{2\pi}\int_{0}^{2\pi}\left|a+be^{i\alpha}\right|^2d\alpha = |a|^2+|b|^2.
\end{equation*}
\end{Lem}

\begin{proof}
We have
\[
   \left|a+be^{i\alpha}\right|^2 =
   |a|^2+\overline{a}be^{i\alpha}+a\overline{b}e^{-i\alpha}+|b|^2.
\]
Integrating the right-hand side over $[0,2\pi]$ gives $2\pi(|a|^2+|b|^2)$.
\end{proof}

\Cref{lem:integral} admits the following higher-dimensional generalisation.
\begin{Lem}\label{lem_integral_unif}
Let $a$ and $b_1$, \dots, $b_n$ (with $n\geq 1$) be complex numbers. One has
\begin{equation*}
  \frac{1}{(2\pi)^n}
  \int_{0}^{2\pi} \dotsi \int_{0}^{2\pi} \left|a+b_1e^{i\alpha_1}+\dotsb+ b_ne^{i\alpha_n}\right|^2 d\alpha_1\cdots d\alpha_n = |a|^2 + |b_1|^2 + \dotsb + |b_n|^2.
\end{equation*}
\end{Lem}

\begin{proof}
We use induction on~$n$; the case $n=0$ is immediate and the induction step follows from \cref{lem:integral}.
\end{proof}

\begin{Thm}
\label{thm:chi-unif}
Assume that $\theta = 0$ and the images $\chi(\phi)$ of the phase function $\chi$ are chosen independently and uniformly random in the interval $[0, 2\pi)$.  Then the probability of obtaining the shift $s$ is independent of $s$ and is given by
\[
     p = \left(\frac{\hr}{R}\right)^2+ \frac{1}{|G|^2}\sum_{\phi\in \hat{G}}\left(1-\left|\frac{\hr}{\hat{f}(\phi)}\right|^2\right)\left|\hat{h}(\phi)\right|^2.
\]
\end{Thm}

\begin{proof}
For all $\phi\in\Ghat$, we set
\[
    F(\phi) = \frac{1}{|G|^{3/2}}\sum_{x\in G}\sqrt{1-\left|\frac{\hr}{\hat{f}(\phi)}\right|^2}\sqrt{1-\left|\frac{f(x)}{R}\right|^2}\phi(x).
\]
By \cref{Thm:approximate_algorithm}, the conditional probability of obtaining the shift~$s$ given $\chi$ equals
\[
    P(\ket{s}\ket{0}\mid\chi) = \left|\frac{\hr}{R} + \sum_{\phi\in\Ghat}F(\phi)e^{i\chi(\phi)}\right|^2.
\]
Using the law of total probability, the probability $P(\ket{s}\ket{0})$ writes
\[
    P(\ket{s}\ket{0}) = \int_\chi P(\ket{s}\ket{0}\mid\chi)d\chi.
\]
We evaluate the latter using \cref{lem_integral_unif} to get 
\begin{align*}
    P(\ket{s}\ket{0}) &= \left(\frac{\hr}{R}\right)^2+\sum_{\phi\in\Ghat}|F(\phi)|^2\\
    & = \left(\frac{\hr}{R}\right)^2 +\frac{1}{|G|^{2}} \sum_{\phi\in\Ghat}\left(1-\left|\frac{\hr}{\hat{f}(\phi)}\right|^2\right)\left|\frac{1}{|G|^{1/2}}\sum_{x\in G}\sqrt{1-\left|\frac{f(x)}{R}\right|^2}\phi(x)\right|^2\\
    & = \left(\frac{\hr}{R}\right)^2 +\frac{1}{|G|^{2}} \sum_{\phi\in\Ghat}\left(1-\left|\frac{\hr}{\hat{f}(\phi)}\right|^2\right)\left|\hat{h}(\phi)\right|^2.\qedhere
\end{align*}
\end{proof}

Finally, we consider the case where both the images $\theta(x)$ of the phase function $\theta$ and the images $\chi(\phi)$ of the phase function $\chi$ are chosen independently and uniformly randomly in the interval $[0,2\pi)$.

\begin{Thm}\label{thm:gene}
Assume that the images $\chi(\phi)$ and $\theta(x)$ of the respective phase functions $\chi$ and $\theta$ are chosen independently and uniformly randomly in the interval $[0, 2\pi)$.  Then the probability of obtaining the shift $s$ is independent of $s$ and is given by
\[
     p = \left(\frac{\hr}{R}\right)^2+ \frac{1}{|G|^{3}}\sum_{(\phi,x)\in(\Ghat,G)}\left(1-\left|\frac{\hr}{\hat{f}(\phi)}\right|^2\right)\left(1-\left|\frac{f(x)}{R}\right|^2\right).
\]
\end{Thm}
\begin{proof}
Note that one can generalize \cref{lem_integral_unif} as follows: if $a$ and $\{b_{j,k}\}_{j,k=1}^n$ are also complex numbers, then
    \begin{equation}\label{eq:new}
  \frac{1}{(2\pi)^{2n}}
  \int_{0}^{2\pi} \dotsi \int_{0}^{2\pi}
  \int_{0}^{2\pi} \dotsi \int_{0}^{2\pi} \left|a+\sum_{j,k}b_{j,k}e^{i\alpha_j}e^{i\beta_k}\right|^2 d\alpha_1\cdots d\alpha_n d\beta_1\cdots d\beta_n= |a|^2 + \sum_{j,k}|b_{j,k}|^2 .
\end{equation}
By the law of total probability, the probability $P(\ket{s}\ket{0})$ is given by
\[
    P(\ket{s}\ket{0}) = \int_\theta\int_\chi P(\ket{s}\ket{0}\mid (\chi, \theta))d\chi d\theta.
\]
By proceeding as before and using \cref{eq:new}, we can infer that
\[
    P(\ket{s}\ket{0}) = \left(\frac{\hr}{R}\right)^2+\sum_{(\phi,x)\in(\Ghat,G)}|F(\phi,x)|^2\\
\]
with
\[
    F(\phi,x) = \frac{1}{|G|^{3/2}}\sqrt{1-\left|\frac{\hr}{\hat{f}(\phi)}\right|^2}\sqrt{1-\left|\frac{f(x)}{R}\right|^2}.\qedhere
\]
\end{proof}

\section{Discussion}\label{sec:Discussion}

\subsubsection*{Improving the success probability via amplitude amplification}

We can use amplitude amplification \cite{BHMT} to increase the `good' part of the output state in our approximate \cref{alg: approx1,alg: approx2,alg: generalised bent multid,alg:approx}, and thus the success probability of finding the hidden shift.
To do so, we assume that we are in one of the cases corresponding to \cref{thm: first gen,thm: second gen,thm: multidim gen,th:lowerbound,thm:chi-unif,thm:gene}, where the success probability $p$ is known.
In the case of \cref{alg: approx2,alg: generalised bent multid}, the algorithms need to be adapted slightly by storing the output of the checks for $x \in A+s$ and $\phi \in \hA$ (after the first application of $O_g$ and $O_{\hf}$, respectively) in separate qubits.  Instead of measuring these qubits and the last two registers, we use them to define the `good' states for the purpose of amplitude amplification.

We denote by $\mathcal{A}$ the unitary operator corresponding to the whole algorithm, excluding measurements.
Let $S_\chi$ be the operator that sends every good basis state $\ket{\psi}$ to $-\ket{\psi}$ and fixes the bad basis states. Let $S_0$ be the operator that sends the initial state $\ket{0}$ to $-\ket{0}$ and fixes all other basis states.
The amplification operator is then given by
\[
    Q = -\mathcal{A}S_0\mathcal{A}\ct S_\chi.
\]
Taking $\theta$ such that $\sin^2(\theta) = p$, we need $\lfloor\pi/4\theta\rfloor=\Theta\left(\frac{1}{\sqrt{p}}\right)$ applications of the amplification operator $Q$, and thus of the operators $\mathcal{A}$ and $\mathcal{A}\ct$, starting with the state $\mathcal{A}\ket{0}$.

\subsubsection*{Oracle for the Fourier transform}

A major weakness of our approach is that in \cref{problem: exact formulation} we require oracle access to the Fourier transform $\hf$ of $f$.
Indeed, in practice it might be expensive to compute $\hf$ from a classical description of $f$, and thus difficult to implement the oracle $O_{\hf}$ for $\hf$.
Nevertheless, in some cases the function $f$ and its Fourier transform are related in a simple way, allowing for an efficient implementation of $O_{\hf}$ from access to $O_f$.
For example, if $f$ is a primitive Dirichlet character or a multiplicative character of a finite field, then $\smash{\hf(y) = \overline{f(y)} \hf(1)}$, see \cref{eq:f hat and f}.
It would be nice if one could find ways of overcoming this limitation more generally.

\subsubsection*{Results on multidimensional bent functions}

In the one-dimensional case, bent functions have been studied in various settings and under various names such as biunimodular vectors, see \cref{rem:bent}. However, not much is known about multidimensional bent functions \cite{poinsot2005multidimensional} defined in \cref{def:multidim bent} and discussed further in \cref{apx:multidim bent}.
In particular, while biunimodular vectors have been classified for small $|G|$ \cite{fuhr2015}, including $|G| = 13$ \cite{Gabidulin2002}, no classification or even non-trivial constructions of multidimensional bent functions are known.

Classification of bent functions for small abelian groups can indicate the `density' of bent functions among all complex functions.
The closer a given function is to a bent function, the higher the success probability of our algorithms.
The more bent functions there are, the more functions are \emph{close} to one, thus having formal results on bent function density would inform us on how well our algorithms work in general.

\subsubsection*{Acknowledgments}
SA thanks the Quantum Software Consortium programme for research stay at Leiden University. PB was partially supported by the Dutch Research Council (NWO/OCW) as part of the Quantum Software Consortium programme (project number 024.003.037). MO was partially supported by the Dutch Research Council (NWO) Vidi grant (No.VI.Vidi.192.109). JS was supported by the Dutch Research Council (NWO) as part of the Quantum Software Consortium programme (project number 024.003.037) and the Quantum Delta NL programme (project number NGF.1623.23.023).

\bibliographystyle{alphaurl}
\bibliography{quantum}

\appendix

\section{Error analysis}\label{sec:error-analysis}

Throughout the article, it was assumed that the oracles 
\[O_g: \ket{x}\ket{0} \to \ket{x}\ket{g(x)}, \: O_{\hf}:\ket{\phi}\ket{0} \mapsto \ket{\phi} \ket{\hf(\phi)}\]
give access to the complex function value $g(x)$ without any approximation needed. In this section, we calculate the error of \cref{alg: approx1} when we take the approximation into account. Specifically, we assume that any number $y\in \C$ can be represented with $2n$ bits by $(\arg(y), |y|),$ where both real numbers are approximated up to precision $\delta=\frac{1}{2^n}$. 

\begin{Prop}
    Let $f: G \to \C$ be an $(r, R, \hr,\hR, A, \hA)$-bounded function and assume complex numbers can be stored up to precision $\delta = \frac{1}{2^n}$ in their argument and modulus. Assume that $n \gg \log(|G|)$. Then there exists a positive real number $C$ such that \cref{alg: approx2} with error finds the hidden shift with probability
    \[
    p(s) - C|G|^{1/2}\delta \leq p_e(s) \leq p(s)+C|G|^{1/2}\delta,
    \]
    where $p(s)$ is the probability without approximation given in \cref{eq: full probability}.
\end{Prop}
\begin{proof}We follow the notation from the proof of \cref{thm: second gen}. By considering oracle access to the functions $g/R$ and $\hr/\hf$, we may assume that all complex numbers have a modulus between $0$ and $1$. This means that invoking the oracle $ O_g $, applying the unitary $U_1$ from equation \cref{eq: U1 second gen} and then applying $O_g$ again introduces an error via 
\begin{align*}
\ket{\Phi(s)} = \frac{1}{|A|^{1/2}|G|^{1/2}}
    &\sum_{\phi \in \smash\hG} \sum_{x \in A+s}
    \phi(x) \ket{\phi,0,0}\left(
        \left(\frac{g(x)}{R} +\epsilon(x)\right)\ket{0}
        + \left(\sqrt{1-\left|\frac{g(x)}{R}\right|^2} +\epsilon'(x)\right)\ket{1}
    \right)
    \ket{0}.
\end{align*}
A similar error holds for $O_{\hf}$ in combination with $U_2$ from \cref{eq: U2 second gen}, which we will denote by $\hepsilon(\phi)$. Reading the circuit of  \cref{alg: approx2} from the right gives the state
\begin{equation*}
    \ket{\Psi(s)}  = \frac{1}{|\hA|^{1/2}}\sum_{\phi \in \smash\hA} \phi(s) \ket{\phi,0,0,0}\left(\left(\frac{\hr}{\overline{\hf(\phi)}}+\hepsilon(\phi)\right)\ket{0} + \left(\sqrt{1-\left|\frac{\hr}{\overline{\hf(\phi)}}\right|^2}+\hepsilon'(\phi)\right)\ket{1}\right).
\end{equation*}
The errors $\epsilon(x)$ and $\hepsilon(\phi)$ are complex number with a small norm. In particular, there is a constant $C_1>0$ such that $|\epsilon(x)|<C_1\delta$ and $|\hepsilon(\phi)|<C_1$ for all $x\in G, \phi \in \hG$.  Following the proof of \cref{thm: second gen} the inner product becomes
\begin{align*}
    \braket{\Psi(s)}{\Phi(s)}&=
    \frac{1}{|\hA|^{1/2}|A|^{1/2}|G|^{1/2}} \sum_{\phi \in \hA}\sum_{x \in A+s} \phi(x)\overline{\phi(s)}\left(\frac{g(x)}{R}+\epsilon(x)\right)\left(\frac{\hr}{\hf(\phi)}+\hepsilon(\phi)\right)\\
    &=\frac{1}{|\hA|^{1/2}|A|^{1/2}|G|^{1/2}} \sum_{\phi \in \hA}\sum_{x \in A+s} \phi(x)\overline{\phi(s)}\left(\frac{g(x)\hr}{\hf(\phi)R}+\hepsilon(\phi)\frac{g}{R}+ \epsilon(x)\frac{\hr}{\hf(\phi)}+\epsilon(x)\hepsilon(\phi)\right)\\
    &=\frac{1}{|\hA|^{1/2}|A|^{1/2}|G|^{1/2}} \sum_{\phi \in \hA}\sum_{x \in A+s} \phi(x)\overline{\phi(s)}\frac{g(x)\hr}{\hf(\phi)R} + \frac{|\hA|^{1/2}|A|^{1/2}}{|G|^{1/2}} O(\epsilon).\\
    &= S + \frac{1}{|\hA|^{1/2}|A|^{1/2}|G|^{1/2}} \sum_{\phi \in \hA}\sum_{x \in A+s} \phi(x)\overline{\phi(s)}\left(\hepsilon(\phi)\frac{g}{R}+ \epsilon(x)\frac{\hr}{\hf(\phi)}+\epsilon(x)\hepsilon(\phi)\right)\\
    &=S + E,
\end{align*}
where $S$ is the sum in \cref{eq:p(s) double sum} obtained in the version without approximation. The double sum denoted by $E$ is a complex number of  which we can estimate the norm. Using the triangle inequality and the fact that $|A|,|\hA| \leq |G|$ as well as $\left|\frac{g(x)}{R}\right|, \left|\frac{\hr}{\hf(\phi)}\right|\leq1$ we obtain
\begin{align*}
    |E| &= \left| \frac{1}{|\hA|^{1/2}|A|^{1/2}|G|^{1/2}}\right. \left.\sum_{\phi \in \hA}\sum_{x \in A+s} \phi(x)\overline{\phi(s)}\left(\hepsilon(\phi)\frac{g}{R}+ \epsilon(x)\frac{\hr}{\hf(\phi)}+\epsilon(x)\hepsilon(\phi)\right) \right| \\
    &\leq \frac{1}{|A|^{1/2}|\hA|^{1/2}|G|^{1/2}}\sum_{\phi \in \hA}\sum_{x \in A+s}\left(|\epsilon(x)|+|\hepsilon(\phi)|+ |\epsilon(x)||\hepsilon(\phi)|\right)\\
    &\leq |G|^{1/2}(2C_1\delta + C_1^2\delta^2) \leq |G|^{1/2}C_2\delta, 
\end{align*}
for some other constant $C_2>0$ under the assumption that $\delta = \frac{1}{2^n}$ is small enough. The total probability to find the hidden shift $s$ with approximation errors is given by 
\[
    p_e(s) = \alpha\halpha\left|S+E\right|^2.
\]
Using the (inverse) triangle inequality we find that 
\begin{align*}
    p_e(s) &\geq \alpha\halpha(|S| - |E|)^2 = p(s) - 2p(s)^{1/2}(\alpha\halpha)^{1/2}|E| + \alpha\halpha|E|^2 \geq p(s) - C|G|^{1/2}\delta,\\
    p_e(s) &\leq \alpha\halpha(|S| + |E|)^2 = p(s) + 2p(s)^{1/2}(\alpha\halpha)^{1/2}|E| + \alpha\halpha|E|^2\leq p(s) + C|G|^{1/2}\delta,
\end{align*}
for some new constant $C>0$.
\end{proof}

\section{Multidimensional bent functions}\label{apx:multidim bent}

In \cref{sec: multi} we presented a quantum algorithm for the hidden shift problem of multidimensional complex bent functions. However, very little is known about the existence and classification of these functions. In this section, we present some examples of multidimensional bent functions and discuss their properties.

The only reference on multidimensional bent functions known to us \cite{poinsot2005multidimensional} provides two constructions:
\begin{enumerate}
    \item The \emph{concatenation construction}. Let $f_i: G\to \C$ with $1\leq i\leq d$ be a family of $1$-dimensional bent functions. Then for any vector $u \in \C^d$ of norm $1$ the function 
    \begin{align*}
        f: G&\to \C^d,\\
        x &\mapsto (u_1 f_1(x), \dots, u_d f_d(x))
    \end{align*}
    is $d$-dimensionally bent.
    \item The \emph{disjoint support construction}. The function $f: G \to \C^{|G|}$ given by $f(g) = e_g$ is a $|G|$-dimensional bent function.
\end{enumerate}
The concatenation construction can be extended to arbitrary dimensions.

\begin{Lem}[Concatenation construction]
    \label{lem: concatenation construction}
    Let $f^{(i)}: G\to \C^{d_i}, \:1\leq i\leq n$ be a collection of bent functions and let $u \in \C^n$ be a vector such that $\norm{u}^2=1$. Then 
    \begin{align*}
    f: G&\to \C^{d_1+\dots+d_n},\\
    x &\mapsto u_1f^{(1)}(x)\oplus\dots\oplus u_nf^{(n)}(x)
    \end{align*}
    is an $(d_1+\dots+d_n)$-dimensional bent function.
\end{Lem}

\begin{proof}
The norm of $f(x)$ is given by 
\[\norm{f(x)}^2 = \sum_{i=1}^n u_i^2\norm{f^{(i)}}^2 = \sum_{i=1}^n u_i^2=1.\]
As the Fourier transform is performed coordinate-wise, the norm of $\hf(\phi)$ is given by 
\[\norm{\hf(\phi)}^2= \sum_{i=1}^n u_i^2 \norm{\hf^{(i)}}^2 = \sum_{i=1}^{n}u_i^2=1,\]
from which we can conclude that $f$ is bent.
\end{proof}

When having the hidden shift problem in mind, the concatenation construction does not give any interesting bent functions. These functions are already bent (up to a scalar) in the first coordinate. The hidden shift can then be found using the one-dimensional algorithm. In fact, the hidden shift problem for many multidimensional bent functions can be reduced to the one-dimensional case.

\begin{Lem}
    Let $G$ be a finite abelian group. The action of $\U{\C^d}$ on the space of $d$-dimensional bent functions $f: G \to \C^d$ given by
    \[
        (Uf)(x) := U \begin{pmatrix}
            f_1(x)\\\vdots\\f_d(x)
        \end{pmatrix}\in \C^d \text{ for all } U\in \U{\C^d}
    \]
    defines an equivalence relation.
\end{Lem}
\begin{proof}
    For a unitary matrix $U$, the norm $\norm{Uf(x)}$ equals $1$ for all $x\in G$. Because
    \[\widehat{Uf}(\phi) = \frac{1}{|G|^{1/2}} \sum_{x\in G} \phi(x)(Uf)(x) = U\frac{1}{|G|^{1/2}} \sum_{x\in G} \phi(x)f(x) = U\hf(\phi)\]
    it holds that $\norm{\widehat{Uf}(\phi)}=1$ for all $\phi\in \hG$, making $Uf$ a bent function. We relate $f$ and $g$ if there is a $U\in \U{\C^d}$ such that $f = Ug$. This is an equivalence relation because $\U{\C^d}$ contains the identity matrix and multiplicative inverses, and is closed under multiplication.
\end{proof}
\begin{Def}[Equivalence of multidimensional bent functions]
    Let $(G,e_G)$ be a finite abelian group and $d\geq 1$ an integer. Define $B_d(G)$ to be the set of equivalence classes of $d$-dimensional bent functions on $G$.
\end{Def}
\begin{Def}[Concatenated multidimensional bent function]
    We call an element $f \in B_d(G)$ \emph{concatenated} if it is equivalent to a bent function obtained from the concatenation construction.
\end{Def}
\begin{Ex}
    Let $G = \Z/3\Z$ and recall the 2-dimensional bent function $f$ given in \cref{Ex: 2-dim bent function}. This function is equivalent to
    \[f': = Uf, \text{ with } U=\frac{1}{\sqrt{2}}\begin{pmatrix}
        1&1\\
        -1&1
    \end{pmatrix},\quad \begin{array}{c|ccc}
    &0&1&2\\
    \hline
    f'_1(x)&\frac{1}{\sqrt{2}}&\frac{\omega}{\sqrt{2}}&\frac{1}{\sqrt{2}}\\
    f'_2(x)&\frac{1}{\sqrt{2}}&\frac{\omega^2}{\sqrt{2}}&\frac{1}{\sqrt{2}}
\end{array}. \]
The new function $f'$ is bent in both coordinates (scaled with $\frac{1}{\sqrt{2}}$) and thus a concatenated bent function.
\end{Ex}
The above example implies that the quantum algorithms in \cref{sec: multi} are not necessary for the hidden shift problem of this function, because applying the correct unitary reduces it to the 1-dimensional case. The following question then arises: 
are all multidimensional bent functions concatenated?

\subsection{Gram matrices}

In order to answer the question posed above, we need to study bent functions up to a suitable notion of equivalence.
As the action of unitary matrices can be seen as a basis change on $\C^d$, it will be easier to study the Gram matrix of a bent function. We first give a key property of multidimensional bent functions and then make an identification with Gram matrices.

\begin{Lem}[\protect{\cite[Proposition~2]{poinsot2005multidimensional}}] \label{lem: multi bent derivative}
    Let $f: G\to \C^d$ be a complex function for which $\norm{f(x)}=1$ for all $x\in G$. Then $f$ is bent if and only if $\sum_{x\in G}f(x)\overline{f(x+a)}=0$ for all $a\neq e_G$.
\end{Lem}
\begin{Def}\label{Def: Gram matrices} Define $C_d(G)$ to be the set of $|G|\times|G|$ matrices $M$ that satisfy
    \begin{enumerate}
        \item $M$ is a hermitian and positive semidefinite matrix of rank $\leq d$,
        \item $M_{x,x} = 1$ for all $x\in G$,
        \item $M$ satisfies $\sum_{x\in G} M_{x,x+a} = 0$ when $a\neq e_G$.    
    \end{enumerate}
\end{Def}
\begin{Prop}
    For a finite abelian group $G$ and an integer $d\geq1$, there is a one-to-one correspondence between $B_d(G)$ and $C_d(G)$.
\end{Prop}
\begin{proof}
    Let $f: G\to\C^d$ be a bent function. Consider the Gram matrix $M$ obtained from the $d$ vectors $(f_i(x))_{x\in G} \in \C^{|G|}$. That is, 
    $$M = F^{\dagger}F \text{, where }F = \begin{pmatrix}
        \dots&f_1(x)&\dots\\
        \dots&f_2(x)&\dots\\
        &\vdots&\\
        \dots&f_d(x)&\dots
    \end{pmatrix}.$$
By construction, $M$ is a Gram matrix of a set of $d$ row vectors and thus hermitian and positive semidefinite. The entries can be written as $M_{xy} = \sum_{i=1}^{d} f_i(x)\overline{f_i(y)}$. Because $f$ is a bent function, it follows that $M_{x,x}=1$ for all $x \in G$. By \cref{lem: multi bent derivative}, the last criterion is also satisfied so that $M \in C_d(G)$. The action of a unitary matrix $U \in \U{\C^d}$ on $f$ changes the column vectors of $F$ to those of $UF$. The Gram matrix then remains
$$(UF)^{\dagger}UF = F^{\dagger}U^{\dagger}UF = F^{\dagger}F = M.$$
We conclude that the map
\[\phi: B_d(G) \to C_d(G): \quad f \mapsto F^{\dagger}F\]
is well defined.

Conversely, any hermitian positive semidefinite matrix of rank $\leq d$ is the Gram matrix $M = F^{\dagger}F$ of a set of $d$ vectors in $\C^{|G|}$ (if the rank is lower than $d$ we can add zeros as rows). We can now construct a function by defining $f$ as the column vectors of $F$. The obtained function $f$ is bent precisely when $M$ satisfies the conditions of \cref{Def: Gram matrices}. If there is another matrix $G$ consisting of $d$ vectors in $\C^{|G|}$ such that $M=G^{\dagger}G$, then it must hold that $F = UG$ for some unitary matrix $U$ \cite[Theorem 7.3.11]{horn2012matrix}. It follows that the map
\[\psi: C_d(G) \to B_d(G), \quad M= F^{\dagger}F \mapsto f\]
is well-defined and the inverse of $\phi$, giving the desired bijection.
\end{proof}

In terms of Gram matrices, concatenation as in \cref{lem: concatenation construction} means the following.
\begin{Coro}
    A bent function $f \in B_{n+m}(G)$ is a concatenation of functions in $B_{n}(G)$ and $B_{m}(G)$ if and only if the corresponding Gram matrix $M$ can be written as $M = tM_1 + (1-t)M_2$ for some $M_1\in C_{n}(G), M_2\in C_m(G)$ and $t\in [0,1]$. In this case we call $M$ \emph{concatenated}.
\end{Coro}

\subsection{The case \texorpdfstring{$G=\Z/3\Z$}{G = Z/3Z}}

For the small cyclic group $\Z/3\Z$ we will show some examples of bent functions in dimension $d=1,2$. First, we classify all one-dimensional bent functions and then we show that not all 2-dimensional ones are equivalent to concatenation constructions. Since $\U{\C} = \{u\in \C: |u|=1\}$ we can rephrase
\begin{align*}
    B_1(\Z/3\Z) &= \{f: \Z/3\Z \to C: f \text{ bent and } f(0)=1\}\\
    &\cong\{f = (1, z_1, z_2)\in \C^3: |z_1| = |z_2| = 1, z_1 + \overline{z_2} + \overline{z_1}z_2 = 0)\} &&\protect{(\text{\cref{lem: multi bent derivative}})}\\
    &=\{f = (1, x_1+iy_1, x_2+iy_2) : x_1^2+y_1^2=x_2^2+y_2^2=1,\\ 
    &\hspace{2cm} x_1+x_2 + x_1x_2 + y_1y_2 = y_1-y_2 +x_1y_2 - x_2y_1 =0\}\\
    &= \{(1,1,\omega), (1, 1, \omega^2), (1, \omega, 1), (1,\omega, \omega^2), (1,\omega^2, 1), (1, \omega^2, \omega)\}. &&(\omega = e^{2\pi i/3})
\end{align*}
Equivalently, 
\[C_1(\Z/3\Z) = \left\{\begin{pmatrix}
    1&a&b\\
    \overline{a}&1&\overline{a}b\\
    \overline{b}&a\overline{b}&1
\end{pmatrix}: (1,a,b) \in B_1  (\Z/3\Z)\right\}\]
and all concatenated matrices in $C_2(G)$ are of the form $tM_1 + (1-t)M_2$ with $t\in [0,1]$ and $M_i\in C_1(\Z/3\Z).$

\begin{Ex}
    Consider the matrix
    \[M:= \begin{pmatrix}
        1&\frac{e^{2\pi i a}}{\sqrt{2}}&-\frac{e^{-2\pi i a}}{\sqrt{2}}\\
        \frac{e^{-2\pi i a}}{\sqrt{2}}&1&0\\
        -\frac{e^{2\pi i a}}{\sqrt{2}}&0&1
    \end{pmatrix}\]
with $a\in \R$. It can be readily checked that $M$ lies in $C_2(\Z/3\Z)$ and it has rank $2$ with eigenvalues $0,1,2$. The corresponding bent function (up to equivalence) is given by 
\[f: \Z/3\Z \to \C^2, \qquad \begin{array}{c|ccc}
    &0&1&2\\
    \hline
    f_1(x)&1&\frac{e^{-2\pi i a}}{\sqrt{2}}&-\frac{e^{2\pi i a}}{\sqrt{2}}\\
    f_2(x)&0&\frac{e^{-2\pi i a}}{\sqrt{2}}&\frac{e^{2\pi i a}}{\sqrt{2}}
\end{array}\;. \]
Using a computer program like \textit{Mathematica} reveals that $M$ is not equal to any combination of two points in $C_1(\Z/3\Z)$, making $f$ not concatenated.
\end{Ex}

\section{Forrelation for abelian groups}\label{app: forrelation}

The forrelation problem for boolean functions was introduced by Aaronson and Ambainis \cite{aaronson2015forrelation} as an example of optimal separation between quantum and classical query complexity. We explain the concept of forrelation and extend it to functions on general abelian groups.

For two boolean functions $f, g: \F_2^n \to \{-1,1\}$, the \emph{forrelation} (or Fourier correlation) is given by 
\begin{equation}
    \Phi_{f,g} := \frac{1}{2^n}\sum_{x \in \F_2^n} f(x) W_g(x) = \frac{1}{2^{3n/2}}\sum_{x,y\in \F_2^n} f(x)(-1)^{x\cdot y}g(y),
\end{equation}
where $W_g$ denotes the Walsh-Hadamard transform of $g$, also denoted by $\hat{g}$. The forrelation problem is then to determine whether $|\Phi_{f,g}|\leq \frac{1}{100}$ or $|\Phi_{f,g}|\geq \frac{3}{5}$, assuming one of the two holds. The relation with boolean bent functions is given by the following lemma.

\begin{Prop}[\protect{\cite[Proposition 5.2]{dutta2024nega-forrelation}}]
    Given bent functions $f, g: \F_2^n \to \{-1,1\}$ such that $g(x) = f(x\oplus u)$ for all $x\in \F_2^n$, then 
    \begin{equation}
        \Phi_{g, \hf} = \begin{cases}
            1 &\text{ if } u=0^n,\\
            0 &\text{ if } u\neq 0^n.
        \end{cases}
    \end{equation}
\end{Prop}
For boolean functions, taking the Fourier transform twice gives the identity. For general abelian groups this is not the case anymore. By the correspondence between a group $G$ and its double dual
\begin{equation}
    G \to \hat{\hG}, \quad x \mapsto \{\psi_x: \phi \mapsto \phi(x) \text{ for all } \phi\in \hat{G}\},
\end{equation}
one can easily check that $\hat{\hf}(\psi_x) = f(-x)$. Instead, we make use of the inverse Fourier transform.
\begin{Def}[Inverse Fourier transform]
    For a function $h: \hG \to \C$, the \emph{inverse Fourier transform} is given by 
    \begin{equation}\label{eq: inverse fourier}
        \check{h}: G  \to \C, \quad \check{h}(x) = \frac{1}{|G|^{1/2}} \sum_{\phi\in \hG} \overline{\phi(x)}h(\phi).
    \end{equation}
    For a multidimensional function the inverse Fourier transform is applied coordinate-wise.
\end{Def}
Implied by its name, the inverse Fourier transform is the inverse of the Fourier transform. We use it to define forrelation for multidimensional functions on abelian groups. 
\begin{Def}[Multidimensional forrelation]
    Let $g: G \to \C^d$ and $h: \hG \to \C^d$ be two functions. The \emph{forrelation} between $g$ and $h$ is given by 
    \begin{equation}
        \Phi_{g,h} = \frac{1}{|G|} \sum_{x\in G}\sum_{i=1}^d g_i(x) \overline{\check{h}_i(x)}.
    \end{equation}
\end{Def}
The forrelation between two bent functions that are shifted relative to each other has the following property.
\begin{Prop}
    Let $f,g: G\to \C^d$ be bent functions hiding a shift $s$. Then 
    \begin{equation}
        \Phi_{g, \hf} = \begin{cases}
            1 &\text{ if } s=0,\\
            0 &\text{ if } s\neq 0.
        \end{cases}
    \end{equation}
\end{Prop}
\begin{proof}
    Writing out the forrelation definition we obtain 
    \[\Phi_{g, \hf} = \frac{1}{|G|} \sum_{x\in G}\sum_{i=1}^d g_i(x) \overline{\check{\hf}_i(x)} = \frac{1}{|G|}\sum_{x\in G}\sum_{i=1}^df_i(x-s)\overline{f_i(x)}.
    \]
    By \cite[Proposition 5.2]{poinsot2005multidimensional} this expression is zero for all $s \neq 0$, from which the result follows.
\end{proof}

\end{document}